\documentclass[12pt]{article}
\usepackage[cp1251]{inputenc}
\usepackage[russian]{babel}
\usepackage{amsmath}
\usepackage{amssymb}
\usepackage{amscd}
\usepackage{amsfonts}
\usepackage[dvips]{graphicx}
\usepackage{wrapfig}
\usepackage{pstricks}
\usepackage{enumerate}
\usepackage{pst-plot}
\usepackage[matrix,arrow,curve]{xy}

\newenvironment{proof}{\noindent{\it Proof.\ }}{\hfill$\square$}

\newcommand{\beq}[1]{\begin{equation}\label{#1}}
\newcommand{\eq}{\end{equation}}
\newcommand{\beqn}[1]{\begin{eqnarray}\label{#1}}
\newcommand{\eqn}{\end{eqnarray}}

\def\sl2{{\rm sl}(2, {\mathbb C})}

\def\f1#1{\frac{1}{#1}}

\def\f1#1{\frac{1}{#1}}

\newtheorem{prop}{Proposition}[section]
\newtheorem{defi}{Definition}[section]
\newtheorem{rem}{Remark}[section]
\newtheorem{cor}{Corollary}[section]
\newtheorem{lem}{Lemma}[section]
\newtheorem{theorem}{Theorem}[section]

\begin{document}
\begin{flushright}
ITEP-TH-45/12
\end{flushright}

\vspace{2cm}
\begin{center}
{\Large{\bf Bruhat Order in Full Symmetric Toda System                                                                       }
}\\
\ \\
Yu.B. Chernyakov\footnote{Institute for Theoretical and Experimental Physics, Bolshaya Cheremushkinskaya, 25,
117218, Moscow, Russia.}$^{,}$\footnote{Joint Institute for Nuclear Research, Bogoliubov Laboratory of Theoretical Physics, 141980, Dubna, Moscow region, Russia.}, chernyakov@itep.ru\\
G.I Sharygin\footnotemark[1]$^{,}$\footnotemark[2]$^{,}$\footnote{Lomonosov
Moscow State University, Faculty of Mechanics and Mathematics, GSP-1, 1 Leninskiye Gory, Main Building, 119991, Moscow, Russia.}$^{,}$\footnote{Laboratory of Discrete and Computational Geometry, Yaroslavl' State University, Sovetskaya st.14, 150000, Yaroslavl', Russia.}, sharygin@itep.ru\\
A.S. Sorin\footnotemark[2], sorin@theor.jinr.ru

\end{center}
\renewcommand{\abstractname}{Abstract}
\renewcommand{\figurename}{Figure}
\renewcommand{\refname}{Literature}
\begin{abstract}
In this paper we discuss some geometrical and topological properties of the full symmetric Toda system. We show by a direct inspection that the phase transition diagram for the full symmetric Toda system in dimensions $n=3,4$ coincides with the Hasse diagram of the Bruhat order of symmetric groups $S_3$ and $S_4$. The method we use is based on the existence of a vast collection of invariant subvarieties of the Toda flow in orthogonal groups. We show how one can extend it to the case of general $n$. The resulting theorem identifies the set of singular points of $\mathrm{dim}=n$ Toda flow with the elements of the permutation group $S_n$, so that points will be connected by a trajectory, if and only if the corresponding elements are Bruhat comparable. We also show that the dimension of the submanifolds, spanned by the trajectories connecting two singular points, is equal to the length of the corresponding segment in the Hasse diagramm. This is equivalent to the fact, that the full symmetric Toda system is in fact a Morse-Smale system.
\end{abstract}

\section{Introduction}
Non-periodic Toda system (Toda chain) consists of $n$ particles on the line with interactions between neighbours.
The Hamiltonian of this system is given by
\beq{H-Toda}
H = \sum^{n}_{i=1}\frac{1}{2}p_{i}^{2} + \sum^{n-1}_{i=1} \exp(q_{i}-q_{i+1}),
\eq
where $p_{i}$ is the momentum of the $i^{th}$ particle and $q_{i}$ is its coordinate. The Poisson structure on the phase space $(p_{i}, q_{i})$ has the well-known form
\beq{Poiss-str-1}
\{ p_{i}, q_{j} \} = \delta_{ij}, \ \{ p_{i}, p_{j} \} = 0, \ \{ q_{i}, q_{j} \} = 0.
\eq
The evolution of the system is given by the standard Hamiltonian equations: $p'_i=\{H,\,q_i\},\ q'_i=-\{H,\,p_i\}$.

If we make the following ansatz ( \cite{F1})
\beq{Var}
b_{i} = p_{i}, \ \ \ a_{i} = \exp \frac{1}{2}(q_{i}-q_{i+1}),
\eq
the Hamiltonian will take the form
\beq{H-Toda-2}
H = \sum^{n}_{i=1}\frac{1}{2}b_{i}^{2} + \sum^{n-1}_{i=1} a_{i}^{2}.
\eq
Observe, that there are only $n-1$ variables $a_i$. The Poisson structure (\ref{Poiss-str-1}) turns into
\beq{Poiss-str-2}
\{ b_{i}, a_{i-1} \} = -a_{i-1}, \ \{ b_{i}, a_{i} \} = a_{i}.
\eq
All the other brackets of coordinates $a_i$ and $b_j$ are equal to zero. In these coordinates it is easy to find the Lax representation of the system: one can show that the Hamiltonian equations are equivalent to the following matrix equation
\beq{LM}
\dot{L} = [L,A]
\eq
where $L$ is called the \textit{Lax operator} and it is given by
\beq{Lax} L = \left(
\begin{array}{c c c c c c}
 b_{1} & a_{1} & 0 & ... &0\\
 a_{1} & b_{2} & a_{2} &... & 0\\
 0 & ... & ... & ...& 0\\
 0 & ... & a_{n-2} & b_{n-1} &a_{n-1}\\
 0 & 0 & ... &  a_{n-1} & b_{n}\\
\end{array}
\right)
\eq
and $A$ is the operator
\beq{Lax} A = \left(
\begin{array}{c c c c c c}
 0 & -a_{1} & 0 & ... &0\\
 a_{1} & 0 & -a_{2} &... & 0\\
 0 & ... & ... & ...& 0\\
 0 & ... & a_{n-2} & 0 & -a_{n-1}\\
 0 & 0 & ... & a_{n-1} & 0\\
\end{array}
\right).
\eq
This system was first considered in \cite{T1, T2}, and in the work \cite{H} there were found $n$ functionally independent integrals of motion. The involution of the integrals of motion was proved in the works \cite{F1, F2}. In particular, it is easy to see that the trace of $L$ is a Casimir function of the system, so that we can set it equal to $0$, without any loss of generality. Thus we can assume, that $Tr(L)=0$.

There is a strightforward way to generalize these formulas: as one sees the Lax matrix of this system is tri-diagonal traceless symmetric matrix and $A$ is the tri-diagonal anti-symmetric matrix, obtained from $L$ by the natural procedure of deleting the principal diagonal and changing the signs of the elements above it. Now we can introduce a new system by taking $L$ to be an arbitrary traceless symmetric matrix and $A$ its anti-symmetrization and defining the dynamics by equation \eqref{LM}. This equation is called the \textit{full symmetric Toda system}. Below (see paragraph \ref{Toda_setting}) we shall describe this system in more detail.

As it follows from equation \eqref{LM}, eigenvalues $\lambda_i$ of $L$ are invariants of the motion. Let us assume, that all these numbers are different. Moser in \cite{Mos} showed that at $t \rightarrow -\infty$ the Lax operator of the usual tri-diagonal Toda lattice converges to the diagonal matrix with eigenvalues put in the increasing order, and when $t \rightarrow +\infty$ it converges to the diagonal matrix with decreasing eigenvalues. This property has been further studied in \cite{DNT}, where a general framework for calculating the eigenvalues of a symmetric matrix using the Toda equations was developed. The physical interpretation of this result is the full momentum transfer between $n$ particles (see, for example, \cite{Mos2}). The analogy of this property emerges in the system of (for example) two balls on a line. One of the balls move with the momentum $p$ another ball has zero momentum. After elastic impact the balls change momentum. This system in the center of mass frame is analogous to the Toda system with the $2 \times 2$ matrix of the Lax operator.

In \cite{KM} it was shown that this property is also true for the full symmetric Toda system and in \cite{FS} the transitions between the stable points of this equation in the case $n=3$ for the full symmetric Toda system were described with the help of the explicit formulas for the generic solutions.

In present paper we study the behavior of the Lax operator of the full symmetric Toda system at $t \rightarrow \pm\infty$. It turns out, that it always converges to a diagonal matrix with fixed eigenvalues, whose order at $\pm\infty$ depends on the choice of the initial data at $t=0$ of the trajectory. More accurately, we prove the following theorem:

\medskip
\noindent\textbf{Theorem \ref{theomag}.}\ \textit{The singular points of the full symmetric Toda flow on $SO_n(\mathbb R)$ can be identified with the elements of the symmetric group $S_n$, so that two points are connected by a trajectory if and only if they are comparable in Bruhat order on $S_n$ (after suitable permutation corresponding to the chosen order of eigenvalues).}

\medskip
In this theorem we use the identification of the group $S_n$ with the group of permutations of the eigenvalues of $L$. Thus, we can describe all possible asymptotics of the Lax operator: it turns out that we cannot obtain any given permutation of the eigenvalues, when $t\to\pm\infty$. In effect, the only permutations that are allowed are those, which go along the edges of the Hasse diagramm of the Bruhat order. Moreover, we can describe the dimension of the space of trajectories, that define such permutation: we prove the following

\medskip
\noindent\textbf{Corollary \ref{corMS}.}\ \textit{The dimension of the space, spanned by the trajectories, connecting two singular points is equal to the distance in the Hasse graph of the Bruhat order between these points.}

\medskip
In fact, we show that the system can be described as a \textit{Morse-Smale system} on the flag manifold, so that its stable (resp. unstable) submanifolds are equal to the dual (resp. ordinary) Schubert cells.

The rest of the paper is organized as follows: in section 2 we give synopsis of the facts from dynamical systems, Morse theory and the Schubert calculus, necessary for understanding the rest of the paper. In section 3, which consists of three subsections, we prove the main theorem. We begin with the simplest cases, $n=3$ and $n=4$ (paragraphs 3.1 and 3.2) and in paragraph 3.3 we prove the main results of this paper (theorem \ref{theomag} and corollary \ref{corMS}). Finally in appendix we describe the topology of the full flag space of $\mathbb R^3$.

\section{Definitions and conventions}
This section contains standard results and definitions from various mathematical theories, which we use in this paper. This section is also the place, where we set notation and describe the conventions we use.

\subsection{Generalised Toda flows}
\label{secttodagen}
In these two subsections we shall recall the basics on the full symmetric Toda flow.

\subsubsection{General results}
\label{Toda_setting}
As we have already mentioned in the introduction, it is possible to generalize the Toda system (tri-diagonal Toda chain) and to consider the full symmetric Toda system. It turns out that this system is in effect a Hamiltonian system: one can regard it as the dynamical system on the orbits of the coadjoint action of the Borel subgroup $B^+_n$ of $SL_n(\mathbb R)$ (equal to the group of upper triangular matrices with determinant $1$) see \cite{A, Ad, K1, S1}. Namely, consider the following decomposition (we use the Killing form on $\mathfrak{sl}_n$ to obtain the identifications in the second and the third lines):
\beq{Decomp-2}
\begin{array}{c}
\mathfrak{sl}_n=\mathfrak{so}_n \oplus \mathfrak{b}^+_n,\\
\mathfrak{sl}^{\ast}_{n} = (\mathfrak{b}^{+}_n)^{\ast} \oplus (\mathfrak{so}_n)^{\ast} \cong Symm_n \oplus \mathfrak{n}_n^{+},\\
(\mathfrak{b}^{+}_n)^{\ast} \cong (\mathfrak{so}_n)^{\perp} = Symm_n, \ \ \ (\mathfrak{so}_n)^{\ast} \cong (\mathfrak{b}^{+}_n)^{\perp} = \mathfrak{n}_n^{+}.
\end{array}
\eq
As one sees, we can identify the space of symmetric matrices with the dual space of Lie algebra of Borel subgroup: $Symm_n\cong(\mathfrak{b}_n^+)^*$, and hence we can introduce a symplectic structure on $Symm_n$, pulling it back from $(\mathfrak{b}_n^+)^*$. We define a Hamiltonian function on $Symm_n$ by the formula $H(L)=Tr(L^2), L \in Symm_n$. Then the full symmetric Toda system is just the Hamiltonian system, associated with this data.

For example, if $n=3$ we shall have
$$
L=\begin{pmatrix}a_{11} & a_{12} & a_{13}\\ a_{12} & a_{22} & a_{23}\\ a_{13} & a_{23} & a_{33}\end{pmatrix}\ \mbox{then}\ A=\begin{pmatrix}0 & a_{12} & a_{13}\\ -a_{12} & 0 & a_{23}\\ -a_{13} & -a_{23} & 0\end{pmatrix},
$$
where $A=A(L)$ is the antisymmetric matrix, determined by $L$: $A(L)=L_+-L_-$ (the above-diagonal part of $L$ minus the under-diagonal part of it).

It turns out, that the Hamiltonian system defined in this way is integrable. Some of the integrals of this system are easy to guess: take $H_k(L)=Tr(L^k),\ k=1,\dots,n$ (so that $H(L)=H_2(L)$), then $\{H_i,\,H_j\}=0$, which is proved by direct computation. However, this set of integrals is not complete (in a generic point): dimension of the Poisson manifold is $\dfrac{n(n+1)}{2}$, and the number of integrals is only $n$; thus one should look for additional invariants. First they were found in \cite{DLNT} using the chopping procedure. Another (geometric) construction of the integrals was proposed by the first and the third authors of this paper in \cite{CS}. Their construction is based on the Pl\"ucker embedding of the flag space into products of projective spaces.

In effect, similar constructions can be considered for arbitrary Cartan pairs (decompositions of a semisiple Lie algebra into a direct sum of its compact sublagebra and a subspace which verify certain properties). For any such decomposition one can associate an integrable system analogous to the Toda chain. In this paper we shall restrict our attention to the usual $SL_n$ case. A reader interested in the general case of arbitrary Lie group, should refer for example to the paper \cite{deMariPedroni}.

The evolution of $L$ determined by the Lax equation \eqref{LM} can be described in the terms of orthogonal group: first of all, the equation suggests that the set of eigenvalues of $L$ does not vary with time (although the order of these values need not be preserved), this is equivalent to the invariance of the traces $Tr(L^k)$. Let $\Lambda=diag(\lambda_1,\dots,\,\lambda_n)$, $\sum\lambda_i=0$ be the matrix of eigenvalues of $L$. Then one can find an orthogonal matrix $\Psi=\Psi(t)$, such that
\beq{Lax-GenToda}
L(t) = \Psi(t) \Lambda \Psi^{-1}(t), \ \Psi(t) \in SO(n, \mathbb{R}).
\eq
The matrix $\Psi$ evolves in time according to the law:
\begin{equation}
\label{eqso3}
\frac{d\Psi}{dt} = - A \Psi,
\end{equation}
where now we put $A=A(\Psi)$ is the composition of $A(L)$ with the expression \eqref{Lax-GenToda}:
\begin{equation}
\label{fieldso3}
A(\Psi) = (\Psi \Lambda \Psi^{-1})_{+} - (\Psi \Lambda \Psi^{-1})_{-}.
\end{equation}
For instance, when $n=3$ we have the following expressions for the matrix elements $a_{ij}$ of $L$ and $A$ in terms of the matrix elements $\psi_{ij}$ of $\Psi$ and the eigenvalues:
\beq{matrixelements}
a_{ij} = \sum_{k=1}^{n} \lambda_{k} \psi_{ik} \psi_{jk}.
\eq
It is worthwhile to write down the additional integral of the system in terms of $\psi_{ij}$ and $\lambda_i$: for $n=3$ there is only one additional integral
$$
I_{1,1} = \lambda_{1}\lambda_{2} \frac{\psi_{13}\psi_{33}}{a_{13}}+\lambda_{1}\lambda_{3} \frac{\psi_{12}\psi_{32}}{a_{13}}+\lambda_{2}\lambda_{3} \frac{\psi_{11}\psi_{31}}{a_{13}}.
$$
One can show, that this function is in fact a Casimir in this case.

In fact, for any $n$ the additional integrals can be chosen to be equal to the rational expressions of suitable minors $M_{\frac{i_1\ldots i_k}{1\ldots k}}$ and $M_{\frac{i_1\ldots i_k}{n-k+1\ldots n}}$ of the matrix $\Psi$. Here $M_{\frac{i_1\ldots i_k}{1\ldots k}}$ (resp. $M_{\frac{i_1\ldots i_k}{n-k+1\ldots n}}$) equal to the determinant of the submatrix of $\Psi$, spanned by the intersections of first (resp. last) $k$ rows and by arbitrary $k$ columns $i_1,\dots,i_k$ of $\Psi$. The equations of motion of these minors have the following form
\beq{minors}
M^{'} = f(\lambda, \psi) M,
\eq
where $f$ is some regular function of $\lambda$ and $\psi$ and $M$ is such a minor (we omit indices for the sake of brevity).

So, the level set $M=0$ of $M$ is an invariant subvariety of the Toda system. We shall sometimes call these surfaces \textit{the minor surfaces of the Toda flow}. For instance in the case $n=3$ there are six such surfaces, which correspond to six $1 \times 1$ submatrices and when $n=4$ there are eight distinct $1 \times 1$ minors and six $2 \times 2$ minors, which give invariant surfaces. Details can be found in \cite{FS} and \cite{CS} for the case of full symmetric matrix of the Lax operator, see also \cite{EFS} where the Hessenberg matrices of the Lax operator are considered.

\subsubsection{Potential function}
\label{Toda_grad}
An important property of the system we discuss here is, that it is gradient on the level of the orthogonal matrices, i.e. there is a function $F(\Psi)$, such that its gradient, with respect to certain Riemannian structure on $SO_n(\mathbb R)$, is equal to $A\Psi$; this function also depends on the matrix of eigenvalues $\Lambda$. The corresponding functions and symmetric forms on $\mathfrak{so}_n$ are described, for instance, in \cite{BBR},\cite{BG}, and in paper \cite{deMariPedroni} the general case of an arbitrary Cartan pair is described. It truns out that in the case we consider here the matrix $A$ can be written down in the form
$$
A=J([L,N]),
$$
where $N$ is a symmetric matrix, determined by the root system of $SL_n(\mathbb R)$. The matrx $N$ is not uniquely defined; for instance, one can take $N=diag(0,1,\dots,n-1)$; the operator $J$ is an invertable symmetric linear operator on the space of antisymmetric matrices $\mathfrak{so}_n$; $J$ acts by the division by $k$ on the $k$\/-th upper- and lower-diagonals of an antisymmetric matrix. The operator $J$ is used in the definition of the Riemannian structure: one puts
$$
\langle A,\,B\rangle_J=\langle A,\,J^{-1}(B)\rangle,
$$
where $A,\,B\in\mathfrak{so}_n$ are arbitrary antisymmetric matrices and $\langle A,\,B\rangle$ is the Killing form on $\mathfrak{so}_n$, i.e. in our case
$$
\langle A,\,B\rangle_J=-Tr(AJ^{-1}(B)).
$$
Using this Riemannian structure we can represent $A(\Psi)\Psi$ as the gradient of the function
$$
F_n(\Psi)=Tr(LN)=Tr(\Psi\Lambda\Psi^{-1}N).
$$
In fact, if we take $\Psi=(1+\Theta)\Psi_0$, where $\Theta\in\mathfrak{so}_n$ is an infinitesimal deformation of $\Psi$, then we can linearize the function $F_n$ near $\Psi_0$:
$$
\begin{aligned}
F_n(\Psi)&=Tr((1+\Theta)\Psi_0\Lambda\Psi_0^{-1}(1-\Theta)N)\\
         &=F(\Psi_0)+Tr(\Theta\Psi_0\Lambda\Psi_0^{-1}N)-Tr(\Psi_0\Lambda\Psi_0^{-1}\Theta N)+o(\Theta)\\
         &=F(\Psi_0)+Tr(\Theta[\Psi_0\Lambda\Psi_0^{-1},N])+o(\Theta)\\
         &=F(\Psi_0)+Tr(\Theta J^{-1}(J([\Psi_0\Lambda\Psi_0^{-1},N])))+o(\Theta)\\
         &=F(\Psi_0)-\langle\Theta,\,J([\Psi_0\Lambda\Psi_0^{-1},N])\rangle_{J}+o(\Theta),
\end{aligned}
$$
whence the result we need follows directly, since by an easy computation $J[\Psi_0\Lambda\Psi_0^{-1},N]=A(\Psi_0)$, i.e. $grad_{\langle,\rangle_J} F_n(\Psi_0)=-A(\Psi_0\Lambda\Psi_0^{-1})\Psi_0$.

It is instructive to write down explicit formulae for $F_n$. For instance, if $n=3$, we can take
$$
N=\begin{pmatrix}-1 & 0 & 0\\
 0 & 0 & 0\\
 0 & 0 & 1\end{pmatrix},
$$
(although this choice of $N$ differs from the one we gave above) and
\beq{F-Morse-2}
\begin{array}{c}

F_3(L)=Tr(LN)=a_{33}-a_{11},\\

\ \\

F_3(\lambda, \psi) = \lambda_{1}(\psi^{2}_{31}-\psi^{2}_{11}) + \lambda_{2}(\psi^{2}_{32}-\psi^{2}_{12}) + \lambda_{3}(\psi^{2}_{33}-\psi^{2}_{13}).

\end{array}
\eq

\subsection{Basics of the Morse theory}
In this section we recall the most basic definitions and results from the Morse theory. We warn the reader at once, that we shall rather need the geometric part of this theory which deals with the local structure of the trajectories of a gradient vector field, and not its topological consequences which constitute the major contribution of this theory to Mathematics.

So let $M$ be a compact smooth $n$\/-manifold without boundary. Recall, that a point $x\in M$ is called \textit{singular} point of a function $f:M\to\mathbb R$, if $df(x)=0$ (here $df$ is the differential of $f$). Then $f$ is called \textit{Morse function}, if the following three conditions hold:
\begin{enumerate}[{\ \ 1)}]
\item $f$ has only discrete set of singular points;
\item the values of $f$ in its singular points are different;
\item all its singular points are non-degenerate.
\end{enumerate}
In effect, the first condition follows from the third, but we shall not need this. Recall, that a singular point $x_0\in M$ of $f$ is called \textit{regular}, or \textit{nondegenerate}, if its Hesse matrix
$$
d^2f(x_0)=\begin{pmatrix}\frac{\partial^2f}{\partial x_1^2} & \frac{\partial^2f}{\partial x_1\partial x_2} & \dots &
												 \frac{\partial^2f}{\partial x_1\partial x_n}\\
                         \frac{\partial^2f}{\partial x_2\partial x_1} & \frac{\partial^2f}{\partial x_2^2} & \dots &
                         \frac{\partial^2f}{\partial x_2\partial x_n}\\
                         \vdots & \vdots & \ddots & \vdots\\
                         \frac{\partial^2f}{\partial x_n\partial x_1} & \frac{\partial^2f}{\partial x_n\partial x_2} & \dots &
                         \frac{\partial^2f}{\partial x_n^2}
          \end{pmatrix}_{|x=x_0}
$$
is nondegenerate in some coordinate system (and hence in any coordinate system).

Hesse matrix is a real symmetric matrix, and as one knows from linear algebra, its only invariant is the number of positive and negative squares in the canonic form of the corresponding quadratic function:
$$
q(\xi_1,\dots,\xi_n)=\sum_{i,j=1}^n\frac{\partial^2f}{\partial x_i\partial x_j}\,\xi_i\xi_j=\sum_{k=1}^n\varepsilon_k\theta_k^2
$$
where $\theta_k$ are new coordinates and $\varepsilon_k=\pm1$. It is this invariant, called \textit{the index} of the point $x_0$ (which is equal to the number of $-1$ in the sequence $(\varepsilon_1,\dots,\varepsilon_n)$) which plays crucial role in the whole Morse theory.

An important stage in the development of the Morse theory is the following \textit{Morse lemma}:
\begin{lem}
\label{Morlem}
Let $f$ be a smooth function on a manifold $M$ and $x_0$ a nondegenerate singular point of $f$ and $k$ its index. Suppose that there exists an open neighborhood $U$ of $x_0$, which contains no other singular points. Then theere exists a coordinate system (possibly on a smaller open neighborhood $V$ of $x_0$), $\theta_1,\dots,\theta_n$, such that $x_0$ has coordinates $(0,\dots,0)$ in this system and
$$
f|_V\equiv f(x_0)-\theta_1^2-\dots-\theta_k^2+\theta_{k+1}^2+\dots+\theta_n^2.
$$
\end{lem}
Let now $g_{ij}$ be a Riemannian structure on $M$. Then one can consider the gradient vector field of the function $f$:
$$
(\nabla f)^i=g^{ij}\frac{\partial f}{\partial x_j},
$$
where $g^{ij}$ is the inverse matrix of $g_{ij}$. It is clear, that the singular points of this field coincide with those of the function (since the gradient is obtained by raising the indices of the differential). Moreover, \textit{the coordinate system $\theta_1,\dots,\theta_n$ can be chosen so that the Riemanian structure $g_{ij}$ is trivial in the point $x_0$} (this result is sometimes referred to as the \textit{Morse-Palais lemma}). Thus we can draw the local picture of the trajectories of this vector field near a singular point (in fact the trajectories will correspond to the trajectories of the corresponding field on a Euclidean space up to an infinitesimal correction term). Thus in dimension $2$ for instance we obtain one of the pictures from figure \ref{figvfields}.
\begin{figure}

\vspace{-2cm}
\begin{center}
\hspace{.5cm}\includegraphics[scale=.8]{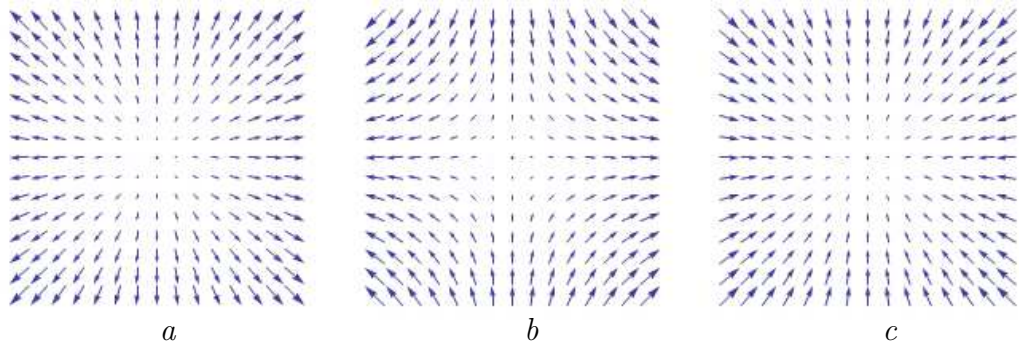}\\

\vspace{-1.8cm}
\hspace{1.1cm}{\it a}\hspace{4.5cm} {\it b} \hspace{4.3cm} {\it c}\\
\end{center}

\caption{\label{figvfields}Types of singularities: {\it a} -- source, {\it b} -- saddle point, {\it c} -- sink.}
\end{figure}

In general we see that the following statement holds
\begin{prop}
\label{Morthtr}
If the index (number of $-1$'s in the canonic form) of a singular point is equal to $k$, then the trajectories that enter this point span (locally) a $k$\/-dimensional submanifold.
\end{prop}

An important particular case of the Morse system is the case of the so-called Morse-Smale systems. Namely, let us call the \textit{stable} (respectively, \textit{unstable}) manifold of a gradient flow at a singular point $p$ of this flow the submanifold, spanned by all the trajectories, outgoing from (resp. incoming into) $p$. These submanifolds depend not only on the function, but also on the Riemannian structure and are usually denoted by $W^s_p(f,g_{ij})$ (resp. $W^u_p(f,g_{ij})$). Then we have the following definition

\begin{defi}
\label{MoSm}
One says that the gradient flow of a Morse function $f$ determines a \textbf{Morse-Smale system}, if all the stable and unstable manifolds of this flow intersect transversely.
\end{defi}

In this case one can describe the dimensions of the intersections of the stable and unstable manifolds in the terms of dimensions of these submanifolds and the dimension of the ambient manifold $M$:
$$
\dim (W^s_p(f,g_{ij})\bigcap W^u_q(f,g_{ij}))=\dim W^s_p(f,g_{ij})+\dim W^u_q(f,g_{ij})-\dim M.
$$
This formula follows from the general results on transversality. This property makes the Morse-Smale systems important tools in algebraic topology.

\subsection{Bruhat order on permutation groups and Schubert cells}
\label{sectorder}
In this section we shall recall the basics on the Bruhat order on the symmetric group and its relation to Schubert cells (for more details see, for example, \cite{F}, \cite{UW}).

Recall that a partial order on a set $X$ is a relation $x<y$ on the elements $x\neq y\in X$, which verify the following condition: if $x<y$ and $y<z$, then $x<z$. One of the convenient ways to visualize a partial order is by its \textit{Hasse diagramm}: it is the oriented graph whose vertices are the elements of the set $X$, and the edges $\overrightarrow{xy}$ correspond to the pairs $x<y$ such that there is no $z$ verifying the relation $x<z<y$. We shall say, that the elements $x$ and $y$ for which this is true are direct neighbours. Observe, that it is enough to know only direct neighbours to restore the whole partial order. The order, which appears in this way, will be called \textit{the minimal order generated by the given set of relations $x<y$}.

Let us now consider the symmetric group $S_{n}$. It consists of all the permutations $w:\{1,2,\dots,n\}\to\{1,2,\dots,n\}$. We shall denote such a permutation by the symbol $w=\begin{pmatrix}1 & 2 & \dots & n\\ i_1 & i_2 & \dots & i_n\end{pmatrix}$ or simply by $w=(i_1,i_2,\dots,i_n)$, which means that $w(k)=i_k$. One can abbreviate this notation to the so-called \textit{cyclic decomposition} of $w$, i.e. one can represent an element $w$ as a composition of non-intersecting cycles $w=(i_1,i_2,\dots,i_{k})(j_1,j_2,\dots,j_l)\dots$, where the symbol $(i_1,i_2,\dots,i_k)$ denotes the permutation $\phi$, such that $\phi(i_1)=i_2,\ \phi(i_2)=i_3,\dots,\phi(i_{k-1})=i_k,\ \phi(i_k)=i_1$, and $\phi(j)=j$ for all $j\neq i_p,\ p=1,\dots,k$ (one should not confuse this decomposition with the previously introduced notation for a permutation). In particular, $(i,j)$ will denote the simplest 2-term cycle which interchanges $i$ and $j$.

One defines the \textit{height} (or \textit{length}) $l(w)$ of a permutation $w$ as the number of inversions in $w$:
\beq{lw}
l(w) = \# \{1\le i<j\le n\mid w(i)>w(j)\}.
\eq
One can now introduce \textit{the (weak) Bruhat order on $S_{n}$} as the minimal partial order, generated by the following relations:
\beq{lw2}
x<y,\ \mbox{if and only if}\ y=(i,j)x,\ \mbox{and}\ l(y)=1+l(x).
\eq

Another approach to the definition of Bruhat order is based on the notion of \textit{Schubert cells} or \textit{Bruhat cells} and \textit{Schubert (Bruhat) varieties} in the flag manifold. Recall, that a (full) flag $E_\cdot$ in $\mathbb R^n$ is a collection of hyperplanes
$$
\{0\}=E_0\subset E_1\subset E_2\subset\dots\subset E_{n-1}\subset E_n=\mathbb R^n
$$
in $\mathbb R^n$, such that $\dim E_i=i$. The set of all flags constitutes a smooth compact manifold $Fl_n(\mathbb R)$ of dimension $\frac{n(n-1)}{2}$ (one can introduce coordinate systems on it from the evident identification $Fl_n(\mathbb R)\cong SL_n(\mathbb R)/B^+_n\cong SO_n(\mathbb R)/(SO_n(\mathbb R)\bigcap B^+_n)$).

For each permutation $w$, one defines the corresponding Schubert cell as the following subset $X_w\in FL_n(\mathbb R)$:
\begin{defi}
\label{defischu1}
Schubert cell $X_{w} \subset Fl_{n}(\mathbb R)$ is the set of the flags verifying certain conditions:
\beq{Xw}
X_{w} = \{E_{\cdot} \in Fl_{n}\ | \ \mathrm{dim}(E_{p} \cap F_{q}) = \# \{ i \mid 1\le i \leq p,\ w(i) \leq q\}\ \forall\ 1 \leq p,q \leq n \},
\eq
where flag $E_{\cdot}= E_{1} \subset E_{2} \subset ... \subset E_{n}=\mathbb{R}^{n}$ and $\mathrm{dim}(E_{i})=i$, $F_{q}= \langle e_{1},...,e_{q} \rangle$ spanned by the first $q$ elements of the basis $e_{1},...,e_{n}$ for $\mathbb{R}^{n}$. We shall denote by $r_w(p,q)$ the number $\# \{ i \mid 1\le i \leq p,\ w(i) \leq q\}$.
\end{defi}
Schubert variety $\bar X_w$ is defined as the closure of the corresponding Schubert cell, which is equivalent to replacing all the equalities for dimensions by inequalities $\ge$. Then one can show that \textit{$w<w'$ in Bruhat order, iff $\bar X_w\subseteq\bar X_{w'}$}, see \cite{F} (in the appendix we write down the equations that determine the preimages in $SO_3(\mathbb R)$ of the Schubert varieties in $Fl_3(\mathbb R)$).

In addition to the Schubert cells $X_w$ one can define \textit{dual Schubert cells} $\Omega_w$, as the sets of those flags $E_\cdot$ for which for all $p,\,q$ we have
$$
\mathrm{dim}(E_p\bigcap\tilde F_q)=\#\{i\le p\mid w(i)\ge n+1-q\},
$$
where $\tilde F_q$ is the subspace, spanned by the last $q$ vectors of the basis. All the statements about the Schubert (or Bruhat) cells and varieties can be transferred to the dual cells and varieties with minimal modifications, see chapter 10 of Fulton's book \cite{F} for details. Below we shall mention those modifications, that we shall need.

There is a more general definition of Bruhat cells (and dual cells) and Bruhat order which works for an arbitrary generalized flag space, i.e. the factor-space of a semi-simple group $G$ by its Borel subgroup $B$. Namely, one can choose an element $w\in G$, representing some element of the Weyl group of $G$, and consider the coset $wB\in G/B$. Then Bruhat cell $X_w$ is equal to the orbit of $wB$ under the left action of $B$ in $G/B$. One can show, that this construction defines a cell decomposition of the flag space, in particular one can show that these orbits are homeomorphic to open discs. The corresponding Bruhat order on Weyl group of $G$ is then given by $v<w$ iff $\bar X_v\subseteq\bar X_w$.

In the case of the usual "geometric"\ flags, i.e. $G=SL_n(\mathbb R)$, the homeomorphism of $X_{w}$ and open disc of dimension $l(w)$ (even $\mathbb{R}^{l(w)}$) can be described in a rather straightforward way (see \S10.2 of \cite{F}): each flag $E_{\cdot}$ has $E_{p}$ spanned by the first $p$ row vectors of a unique "row echelon"\ matrix, the $p^{th}$ row of which has $1$ in the $w(p)^{th}$ column with all 0's after and below these 1's. For instance if $w=(2143)$ such matrices have the form
$$
\tilde\alpha=\begin{pmatrix}* & 1 & 0 & 0\\
                            1 & 0 & 0 & 0\\
                            0 & 0 & * & 1\\
                            0 & 0 & 1 & 0\end{pmatrix},
$$
where $*$ denotes arbitrary real numbers. The corresponding Schubert variety however need not look like a submanifold, but one can describe this variety as the common zero set of a finite collection of polynomial equations on matrices in $SL_n(\mathbb R)$ (or in $SO_n(\mathbb R)$), representing the flags.

For example, consider a matrix $A\in SL_n(\mathbb R)$. As we mentioned before, this matrix determines a flag $E_\cdot$, so that $E_p$ is spanned by the first $p$ row-vectors of $A$ (in some books column-vectors of $A$ are considered). Let $F_\cdot$ be the standard flag, i.e. the flag spanned by rows of the unit matrix. In this case the dimension of the intersection $E_p\bigcap F_q$ is equal to the difference $p-\mathrm{rk}\,\tilde A_{p,q}$, where $\tilde A_{p,q}$ is the submatrix spanned by the first $p$ rows and the last $n-q$ columns of $A$. On the other hand $p-r_w(p,q)$ is equal to the rank of the similar submatrix of the permutation matrix $W\in O_n(\mathbb R)$, associated with $w$. Thus we obtain the following alternative definition of the Schubert cell:
\begin{defi}
\label{defischu2}
Schubert variety (closure of the Schubert cell) associated with a permutation $w\in S_n$ is equal to the subset of flags in $Fl^+_{n}(\mathbb R)=SL_n(\mathbb R)/B_n^+$, spanned by all the matrices $A\in SL_{n}(\mathbb R)$ such that all their upper-right rectangular submatrices have ranks less or equal to the rank of the corresponding submatrix of $W\in O_{n}(\mathbb R)$, the orthogonal matrix associated with the permutation $w$. For dual Schuber cell one should compare the ranks of the upper-left submatrices.
\end{defi}
More accurately the equivalence of these two definitions (\ref{defischu1} and \ref{defischu2}) is proved in \cite{F, UW}.

For instance, take permutation $w=\begin{pmatrix}1 & 2 & 3 & 4\\ 2 & 1 & 4 & 3\end{pmatrix}$, then
$$
W=\begin{pmatrix}0 & 1 & 0 & 0\\
                            1 & 0 & 0 & 0\\
                            0 & 0 & 0 & 1\\
                            0 & 0 & 1 & 0\end{pmatrix},
$$
and we can consider the matrix of ranks, i.e. the matrix, whose $i,j$\/th entry is equal to the rank of the rectangular submatrix $W_{i,j}$ spanned by the first $i$ columns and $j$ rows of $W$:
$$
\mathrm{rk}(W)=\begin{pmatrix}0 & 1 & 1 & 1\\
                                1 & 2 & 2 & 2\\
                                1 & 2 & 2 & 3\\
                                1 & 2 & 3 & 4\end{pmatrix}.
$$
Now the condition that a matrix $A$ defines a flag from the dual Schubert variety $\bar\Omega_w$ is obtained by comparing the ranks of the submatrices in $A$ and the entries of this matrix, as we described above. Thus for this to be true, we should require that the ranks of certain submatrices in $A$ should be less than what is prescribed by their sizes, which is equivalent to vanishing of all the minors, associated with these rectangular submatrices. For instance in the example we consider here, we have the condition $a_{11}=0,\ \det(a_{pq})_{1\le p,q\le3}=0$.

In general, we obtain the polynomial conditions, which would say that all minors of prescribed sizes in a matrix $A$ vanish. Thus the Schubert variety is indeed a variety (on the level of matrices) and we have a complete set of equations, describing it (although the system we obtain in this manner is excessive).

We conclude this section with two general facts about the Schubert cells and dual Schubert cells:
\begin{enumerate}
\item Schubert cell $X_w$ and dual Schubert cell $\Omega_u$ intersect, if and only if $u<w$ in Bruhat order, in which case they intersect transversally;
\item Tangent space to $X_w$ at $w$ is spanned by those positive roots of $SL_n(\mathbb R)$, which are mapped into negative roots by $w$; tangent space to the dual Schubert cell $\Omega_w$ at $w$ is spanned by those positive roots of $SL_n$, which are mapped into positive roots by $w$ (here we consider $w$ as an element of the Weyl group of $SL_n(\mathbb R)$ and identify the tangent space to the flag manifold $SL_n(\mathbb R)/B^+_n$ with the image of the tangent space to $SL_n(\mathbb R)$ under the natural projection).
\end{enumerate}
Both these statements are well-known (see for example \cite{F}, \cite{SBVL}).

\section{The phase portrait of the Toda flow}
In this section we shall give a description of the asymptotic behaviour of the phase trajectories of Toda system for low dimensions. The methods we use are rather simple-minded: knowing that the system is a gradient one, in fact even a Morse system, we see that every its trajectory should flow from one singular point of the vector field to another (in particular, there are no periodic trajectories). Besides this, for each singular point of the vector field we can describe its local space of incoming and outgoing trajectories. Then we use a suitable set of invariant hypersurfaces to describe the combinatorics of this transfer.

\subsection{The case of $SO_3(\mathbb R)$}
\label{sectso33}
We begin with the accurate description of the phase structure of the full symmetric Toda flow for $n=3$. First of all we shall fix a matrix of eigenvalues
$$
\Lambda=\begin{pmatrix}\lambda_1 &0 & 0\\ 0 &\lambda_2 &0\\ 0 & 0& \lambda_3\end{pmatrix},\ \lambda_1+\lambda_2+\lambda_3=0.
$$
We shall assume that $\lambda_i$ are all distinct, more precisely that
$$
\lambda_{2} > \lambda_{1} >0> \lambda_{3}.
$$
Choosing this order is equivalent to the choice of the positive Weyl chamber

As we know, the evolution of the system can be now described in the terms of evolution of the corresponding system on $SO_3(\mathbb R)$, see \eqref{eqso3}. It is easy to see, that the vector field $A(\Psi)$ (see \eqref{fieldso3}) vanishes, if and only if the matrix $\Psi\Lambda\Psi^{-1}$ is diagonal. Since the eigenvalues of $\Lambda$ are all distinct and coincide with eigenvalues of $\Psi\Lambda\Psi^{-1}$, and the matrix $\Psi$ is orthogonal, we conclude that $\Psi\Lambda\Psi^{-1}$ can be diagonal, if and only if $\Psi$ is a matrix with positive determinant that permutes the vectors $\pm e_1,\,\pm e_2,\,\pm e_3$ where $e_1,\,e_2,\,e_3$ is the standard base of $\mathbb R^3$. Such matrices form a finite subgroup $\widetilde S_3$ of $SO_3(\mathbb R)$. Here we give a list of such matrices in $SO_3(\mathbb R)$, symbols $\tilde s_i$ denote the cosets of the elements of $\widetilde S_3$ modulo the group $N=\tilde s_1$ of diagonal matrices in $SO_3(\mathbb R)$, see below:\\
\begin{align*}
\tilde{s_{1}} :
{}&\begin{pmatrix}
                1 & 0 & 0 \\
                0 & 1 & 0 \\
                0 & 0 & 1 \end{pmatrix}, &
{}&\begin{pmatrix}
                -1 & 0 & 0 \\
                0 & -1 & 0 \\
                0 & 0 & 1 \end{pmatrix}, &
{}&\begin{pmatrix} -1 & 0 & 0\\
                0 & 1 & 0 \\
                0 & 0 & -1 \end{pmatrix}, &
{}&\begin{pmatrix} 1 & 0 & 0\\
                0 & -1 & 0 \\
                0 & 0 & -1 \end{pmatrix}, \\
                \tilde{s_{2}} :
{}&\begin{pmatrix} 0 & 1 & 0\\
                -1 & 0 & 0 \\
                0 & 0 & 1 \end{pmatrix}, &
{}&\begin{pmatrix} 0 & -1 & 0\\
                1 & 0 & 0 \\
                0 & 0 & 1 \end{pmatrix}, &
{}&\begin{pmatrix} 0 & 1 & 0\\
                1 & 0 & 0 \\
                0 & 0 & -1 \end{pmatrix}, &
{}&\begin{pmatrix} 0 & -1 & 0\\
                -1 & 0 & 0 \\
                0 & 0 & -1 \end{pmatrix},\ \\
                \tilde{s_{3}} :
{}&\begin{pmatrix} 0 & 0 & 1\\
                0 & 1 & 0 \\
                -1 & 0 & 0 \end{pmatrix}, &
{}&\begin{pmatrix} 0 & 0 & 1\\
                0 & -1 & 0 \\
                1 & 0 & 0 \end{pmatrix}, &
{}&\begin{pmatrix} 0 & 0 & -1\\
                0 & 1 & 0 \\
                1 & 0 & 0 \end{pmatrix}, &
{}&\begin{pmatrix} 0 & 0 & -1\\
                0 & -1 & 0 \\
                -1 & 0 & 0 \end{pmatrix},
\end{align*}
\beq{classesS3}
\eq
\begin{align*}
                \tilde{s_{4}} :
{}&\begin{pmatrix} -1 & 0 & 0\\
                0 & 0 & 1 \\
                0 & 1 & 0 \end{pmatrix}, &
{}&\begin{pmatrix} 1 & 0 & 0\\
                0 & 0 & 1 \\
                0 & -1 & 0 \end{pmatrix}, &
{}&\begin{pmatrix} 1 & 0 & 0\\
                0 & 0 & -1 \\
                0 & 1 & 0 \end{pmatrix}, &
{}&\begin{pmatrix} -1 & 0 & 0\\
                0 & 0 & -1 \\
                0 & -1 & 0 \end{pmatrix}, \\
                \tilde{s_{5}} :
{}&\begin{pmatrix} 0 & 1 & 0\\
                0 & 0 & 1 \\
                1 & 0 & 0 \end{pmatrix}, &
{}&\begin{pmatrix} 0 & -1 & 0\\
                0 & 0 & 1 \\
                -1 & 0 & 0 \end{pmatrix}, &
{}&\begin{pmatrix} 0 & -1 & 0\\
                0 & 0 & -1 \\
                1 & 0 & 0 \end{pmatrix}, &
{}&\begin{pmatrix} 0 & 1 & 0\\
                0 & 0 & -1 \\
                -1 & 0 & 0 \end{pmatrix},\\
                \tilde{s_{6}} :
{}&\begin{pmatrix} 0 & 0 & 1\\
                1 & 0 & 0 \\
                0 & 1 & 0 \end{pmatrix}, &
{}&\begin{pmatrix} 0 & 0 & 1\\
                -1 & 0 & 0 \\
                0 & -1 & 0 \end{pmatrix}, &
{}&\begin{pmatrix} 0 & 0 & -1\\
                1 & 0 & 0 \\
                0 & -1 & 0 \end{pmatrix}, &
{}&\begin{pmatrix} 0 & 0 & -1\\
                -1 & 0 & 0 \\
                0 & 1 & 0 \end{pmatrix}.
\end{align*}
As one sees, $N$ is equal to the intersection of $SO_3(\mathbb R)$ with the group of upper triangular matrices $B^+_n$; this means that the factor-space $SO_3(\mathbb R)/N$ is homeomorphic to the flag space $Fl_3(\mathbb R)$.

We are about to describe the asymptotic behavior of trajectories of the vector field $A(\Psi)$ on $SO_3(\mathbb R)$ (or on $Fl_3(\mathbb R)$), i.e. describe the trajectories that go from one singular point of this field to another. To this end we shall use two important facts about this vector field. First, we recall that this field is the gradient of function $F_3(\Psi)$, see section \ref{Toda_grad} (we can regard this function as a function on $Fl_3(\mathbb R)$, since its value does not depend on the action of $N$). Thus we can assume all the standard properties of the gradient vector fields, in particular, we at once see that there are no closed trajectories of this field since the value of the function $F_3$ always increases along a trajectory. Below we shall show that all the singular points of the field are non-degenerate. Hence we can use the results of Morse theory as well.

Second, one of the important properties of the Toda flow is the fact, that there is a vast collection of invariant surfaces of this vector field, i.e. of the surfaces tangent to the field $A(\Psi)$. These surfaces are described in section \ref{Toda_setting}. In the case $n=3$ there are $6$ distinct surfaces, corresponding to $6$ entries of the first and the third row of $\Psi$ which can vanish. We shall denote these surfaces just by $\psi_{ij}$. Observe, that these equations are also invariant under the action of the group $N$.

So let us begin with the Morse theory.
\begin{prop}
All the singular points of the system are non-degenerate, their index is equal to the height of the corresponding element in the Weyl group (with respect to the chosen positive Weyl chamber).
\end{prop}
\begin{proof}
This statement is proved in a general case in the paper \cite{deMariPedroni}. Here, the simplest way to show this is by a direct inspection of all $6$ points $\tilde{s_i},\ i=1,\dots,6$, listed above (we can choose arbitrary representative in the corresponding coset for our purposes).

\begin{align*}
\tilde{s_{1}}:\begin{pmatrix} \ast & 0 & 0\\
                0 & \ast & 0 \\
                0 & 0 & \ast \end{pmatrix}, \ \ \
\tilde{s_{2}}:\begin{pmatrix} 0 & \ast & 0\\
                \ast & 0 & 0 \\
                0 & 0 & \ast \end{pmatrix}, \ \ \
\tilde{s_{3}}:\begin{pmatrix} 0 & 0 & \ast\\
                0 & \ast & 0 \\
                \ast & 0 & 0 \end{pmatrix}, \\
\tilde{s_{4}}:\begin{pmatrix} \ast & 0 & 0\\
                0 & 0 & \ast \\
                0 & \ast & 0 \end{pmatrix}, \ \ \
\tilde{s_{5}}:\begin{pmatrix} 0 & \ast & 0\\
                0 & 0 & \ast \\
                \ast & 0 & 0 \end{pmatrix}, \ \ \
\tilde{s_{6}}:\begin{pmatrix} 0 & 0 & \ast\\
                \ast & 0 & 0 \\
                0 & \ast & 0 \end{pmatrix},
\end{align*}
where $\ast$ takes the values $\pm1$.\\

It is convenient to choose in a neighbourhood of a point $\Psi\in SO_3(\mathbb R)$ the coordinate system, induced from the Lie algebra $\mathfrak{so}_3$ in the neighbourhood of the unit and transferred to $\Psi$ by the right translation $R_\Psi$. Let $\theta_1,\,\theta_2,\,\theta_3$ be the coordinates in $\mathfrak{so}_3$, i.e. let
$$
\Theta=\begin{pmatrix}
0 & \theta_{1} & \theta_{3}\\
 -\theta_{1} & 0 & \theta_{2}\\
 -\theta_{3} & -\theta_{2} & 0
\end{pmatrix}
$$
be a generic element of $\mathfrak{so}_3$. In order to express the quadratic part of $F_3(\Psi)$ near the singular points $\tilde{s_i},\ i=1,\dots,6$ in coordinates given by translations of $(\theta_1,\,\theta_2,\,\theta_3)$, we replace the matrix $\Psi$ in a neighborhood of $\tilde{s_i}$ by
\beq{decomppsi}
\Psi=\tilde{s_i}+\Theta\tilde{s_i}+\frac12\Theta^2\tilde{s_i}+o(\Theta^2),
\eq
where by $\tilde{s_i}$ we denote some matrix of the class $\tilde{s_i}$ and $o(\Theta^2)$ stands for the sum of terms of degree 3 and higher in $\theta_i$. Thus, we have
\beq{decomplax}
L(\Psi)=\Psi\Lambda\Psi^T=\tilde{s_i}\Lambda\tilde{s_i}^{-1}+ [\Theta, \tilde{s_i}\Lambda\tilde{s_i}^{-1}] - \Theta \tilde{s_i}\Lambda\tilde{s_i}^{-1} \Theta + \frac{1}{2}[\Theta^{2}, \tilde{s_i}\Lambda\tilde{s_i}^{-1}]_{+}+o(\theta^2).
\eq
Thus we can express the quadratic approximation of $F_3$ near $\tilde{s_i}$ in coordinates $(\theta_1,\,\theta_2,\,\theta_3)$ by straightforward calculations. We have gathered the corresponding results in the following table:
\begin{equation*}
\footnotesize{
\begin{tabular}{|c|c|c|c|c|}
\hline\cline{1-0}
$$ & $$ & $$ & $$ & $$\\
$i$ & $F_3(\tilde{s_i})$ & $d^2_{\tilde{s_i}}F_3$ & Index & Minors\\

$$ & $$ & $$ & $$ & $$\\
\hline\cline{1-0}
$$ & $$ & $$ & $$ & $$\\
$1$ & $\lambda_{3}-\lambda_{1}$ & $(\lambda_{1} - \lambda_{2})\theta^{2}_{1} + (\lambda_{2} - \lambda_{3})\theta^{2}_{2} + 2(\lambda_{1}-\lambda_{3})\theta^{2}_{3}$ & $1$ & $\psi_{12},\psi_{13},\psi_{31},\psi_{32}$\\

$$ & $$ & $$ & $$ & $$\\
\hline\cline{1-0}
$$ & $$ & $$ & $$ & $$\\
$2$ & $\lambda_{3}-\lambda_{2}$ & $(\lambda_{2} - \lambda_{1})\theta^{2}_{1} + (\lambda_{1} - \lambda_{3})\theta^{2}_{2} + 2(\lambda_{2}-\lambda_{3})\theta^{2}_{3}$ & $0$ & $\psi_{11},\psi_{13},\psi_{31},\psi_{32}$\\

$$ & $$ & $$ & $$ & $$\\
\hline\cline{1-0}
$$ & $$ & $$ & $$ & $$\\
$3$ & $\lambda_{1}-\lambda_{3}$ & $(\lambda_{3} - \lambda_{2})\theta^{2}_{1} + (\lambda_{2} - \lambda_{1})\theta^{2}_{2} + 2(\lambda_{3}-\lambda_{1})\theta^{2}_{3}$ & $2$ & $\psi_{11},\psi_{12},\psi_{32},\psi_{33}$\\

$$ & $$ & $$ & $$ & $$\\
\hline\cline{1-0}
$$ & $$ & $$ & $$ & $$\\
$4$ & $\lambda_{2}-\lambda_{1}$ & $(\lambda_{1} - \lambda_{3})\theta^{2}_{1} + (\lambda_{3} - \lambda_{2})\theta^{2}_{2} + 2(\lambda_{1}-\lambda_{2})\theta^{2}_{3}$ & $2$ & $\psi_{12},\psi_{13},\psi_{31},\psi_{33}$\\

$$ & $$ & $$ & $$ & $$\\
\hline\cline{1-0}
$$ & $$ & $$ & $$ & $$\\
$5$ & $\lambda_{1}-\lambda_{2}$ & $(\lambda_{2} - \lambda_{3})\theta^{2}_{1} + (\lambda_{3} - \lambda_{1})\theta^{2}_{2} + 2(\lambda_{2}-\lambda_{1})\theta^{2}_{3}$ & $1$ & $\psi_{11},\psi_{13},\psi_{32},\psi_{33}$\\

$$ & $$ & $$ & $$ & $$\\
\hline\cline{1-0}
$$ & $$ & $$ & $$ & $$\\
$6$ & $\lambda_{2}-\lambda_{3}$ & $(\lambda_{3} - \lambda_{2})\theta^{2}_{1} + (\lambda_{1} - \lambda_{2})\theta^{2}_{2} + 2(\lambda_{3}-\lambda_{2})\theta^{2}_{3}$ & $3$ & $\psi_{11},\psi_{12},\psi_{31},\psi_{33}$\\

$$ & $$ & $$ & $$ & $$\\
\hline
\end{tabular}
\
}
\end{equation*}
The second column of this table contains the values of $F_3$ at the given point, the third column consists of quadratic parts of the Taylor series at $\tilde{s_i}$. The nondegenracy of all singular points is now quite evident. Moreover, it follows from the choice of the order of $\lambda_i$, that the values of $F_3$ in $\tilde{s_i}$ are all different
\beq{F-Morse-3}
\begin{array}{c}
F_{3}(\tilde{s_6}) > F_3(\tilde{s_3}) > F_3(\tilde{s_4}) > F_3(\tilde{s_5}) > F_3(\tilde{s_1}) > F_3(\tilde{s_2}).\\
\end{array}
\eq
So the function $F_3$ is Morse, the fourth column of the table contains the Morse index (the number of negative terms in the canonic form of the hessian) of the points.
\end{proof}

In the last column of this table we have given the list of all minor surfaces (see section \ref{Toda_setting}), containing the corresponding singular point. As one knows, every trajectory of a gradient vector field on a compact manifold connects (when the time tends to $\pm\infty$) two different singular points. In addition, all the singular points being nondegenerate, we can use the Morse theory to describe (in terms of the local coordinates) the submanifolds spanned by the incoming and outgoing trajectories (they correspond to the negative and positive eigenspaces of the Hessian matrix). On the other hand, we know, that once in a minor surface a trajectory cannot leave it.

Let now $T_{\tilde{s_k}}\psi_{ij}$ be the tangent space of the minor surface $\psi_{ij}$ at $\tilde{s_k}$ and $T^\pm_{\tilde{s_k}}\psi_{ij}$ its intersection with the positive/negative eigenspace of $d^2_{\tilde{s_k}}F_3$ respectively. We conclude, that the outgoing trajectories that are tangent to $T^+_{\tilde{s_k}}\psi_{ij}$, should connect $\tilde{s_k}$ with another singular point in $\psi_{ij}$, and the trajectories tangent to $T^-_{\tilde{s_k}}\psi_{ij}$ should come to $\tilde{s_k}$ from a point in $\psi_{ij}$. Thus if we can cover the positive and negative eigenspaces inside the tangent space $T_{\tilde{s_k}}SO_3(\mathbb R)$ with $T^\pm_{\tilde{s_i}}\psi_{ij}$, we shall be able to show in what minor surface every incoming/outgoing trajectory of $\tilde{s_k}$ lies.

Besides this we know, that every trajectory moves from lower values of $F_3$ to higher values of this function. Now suppose, that $\tilde{s_l}$ has only one outgoing trajectory inside $\psi_{ij}$ and the only singular point above it in $\psi_{ij}$ is $\tilde{s_k}$. Then we must conclude, that this only trajectory connects $\tilde{s_l}$ and $\tilde{s_k}$ and we don't even need to know, if the surface $\psi_{ij}$ is singular or not. Combining these two simple observations, we can draw the full portrait of the system.

We are going to complete this program. Let us begin with the following table, showing the incidence of singular points and minor surfaces.
\beq{t2}
\footnotesize{
\begin{tabular}{|c|c|c|c|c|c|c|}
\hline\cline{1-0}
$$ & $$ & $$ & $$ & $$ & $$ & $$\\
$i$ & $\psi_{11}$ & $\psi_{12}$ & $\psi_{13}$ & $\psi_{31}$ & $\psi_{32}$ & $\psi_{33}$\\

$$ & $$ & $$ & $$ & $$ & $$ & $$\\
\hline\cline{1-0}
$$ & $$ & $$ & $$ & $$ & $$ & $$\\
$\tilde{s_6}$ & $+$ & $+$ & $$ & $+$ & $$ & $+$\\

$$ & $$ & $$ & $$ & $$ & $$ & $$\\
\hline\cline{1-0}
$$ & $$ & $$ & $$ & $$ & $$ & $$\\
$\tilde{s_3}$ & $+$ & $+$ & $$ & $$ & $+$ & $+$\\

$$ & $$ & $$ & $$ & $$ & $$ & $$\\
\hline\cline{1-0}
$$ & $$ & $$ & $$ & $$ & $$ & $$\\
$\tilde{s_4}$ & $$ & $+$ & $+$ & $+$ & $$ & $+$\\

$$ & $$ & $$ & $$ & $$ & $$ & $$\\
\hline\cline{1-0}
$$ & $$ & $$ & $$ & $$ & $$ & $$\\
$\tilde{s_5}$ & $+$ & $$ & $+$ & $$ & $+$ & $+$\\

$$ & $$ & $$ & $$ & $$ & $$ & $$\\
\hline\cline{1-0}
$$ & $$ & $$ & $$ & $$ & $$ & $$\\
$\tilde{s_1}$ & $$ & $+$ & $+$ & $+$ & $+$ & $$\\

$$ & $$ & $$ & $$ & $$ & $$ & $$\\
\hline\cline{1-0}
$$ & $$ & $$ & $$ & $$ & $$ & $$\\
$\tilde{s_2}$ & $+$ & $$ & $+$ & $+$ & $+$& $$\\

$$ & $$ & $$ & $$ & $$ & $$ & $$\\
\hline
\end{tabular}
}
\eq
The plus sign signifies, that the corresponding point belongs to the surface. We can also describe the restrictions of the quadratic Hessian forms to the tangent space of the minor surfaces. To this end we should consider the linear relation on $(\theta_1,\,\theta_2,\,\theta_3)$ arising from the equation $\psi_{ij}=0$ at a point $\tilde{s_i}$ and restrict $d^2_{\tilde{s_i}}F_3$ to the corresponding subspace. The linear parts have the following form:
\begin{align*}
\Theta \cdot \tilde{s_{1}}:\begin{pmatrix} 0 & \theta_{1}\cdot\ast & \theta_{3}\cdot\ast\\
                -\theta_{1}\cdot\ast & 0 & \theta_{2}\cdot\ast\\
                -\theta_{3}\cdot\ast & \theta_{2}\cdot\ast & 0 \end{pmatrix}, \ \ \
\Theta \cdot \tilde{s_{2}}:\begin{pmatrix} \theta_{1}\cdot\ast & 0 & \theta_{3}\cdot\ast\\
                0 & -\theta_{1}\cdot\ast & \theta_{2}\cdot\ast \\
                -\theta_{2}\cdot\ast & -\theta_{3}\cdot\ast & 0 \end{pmatrix}, \\
\Theta \cdot \tilde{s_{3}}:\begin{pmatrix} \theta_{3}\cdot\ast & \theta_{1}\cdot\ast & 0\\
                \theta_{2}\cdot\ast & 0 & -\theta_{1}\cdot\ast\\
                0 & -\theta_{2}\cdot\ast & -\theta_{3}\cdot\ast \end{pmatrix}, \ \ \
\Theta \cdot \tilde{s_{4}}:\begin{pmatrix} 0 & \theta_{3}\cdot\ast & \theta_{1}\cdot\ast\\
                -\theta_{1}\cdot\ast & \theta_{2}\cdot\ast & 0\\
                -\theta_{3}\cdot\ast & 0 & -\theta_{2}\cdot\ast \end{pmatrix}, \\
\Theta \cdot \tilde{s_{5}}:\begin{pmatrix} \theta_{3}\cdot\ast & 0 & \theta_{1}\cdot\ast\\
                \theta_{2}\cdot\ast & -\theta_{1}\cdot\ast & 0\\
                0 & -\theta_{3}\cdot\ast & -\theta_{2}\cdot\ast \end{pmatrix}, \ \ \
\Theta \cdot \tilde{s_{6}}:\begin{pmatrix} \theta_{1}\cdot\ast & \theta_{3}\cdot\ast & 0\\
                0 & \theta_{2}\cdot\ast & -\theta_{1}\cdot\ast\\
                -\theta_{2}\cdot\ast & 0 & -\theta_{3}\cdot\ast \end{pmatrix},
\end{align*}
The equations of minor surfaces being simply $\psi_{ij}=0$, the corresponding linear relations take a particularly simple form $\theta_k=0$ for certain $k=1,2$ or $3$; thus the restriction of $d^2_{\tilde{s_i}}F_3$ is given simply by removing this coordinate. So we sum up the results in the following table (we replace the actual coefficients at the squares of $\theta_i$ by their signs for the sake of brevity)
\beq{t3}
\footnotesize{
\begin{tabular}{|c|c|c|c|c|c|c|}
\hline\cline{1-0}
$$ & $$ & $$ & $$ & $$ & $$ & $$\\
$\tilde{s_i}$ & $\psi_{11}$ & $\psi_{12}$ & $\psi_{13}$ & $\psi_{31}$ & $\psi_{32}$ & $\psi_{33}$\\

$$ & $$ & $$ & $$ & $$ & $$ & $$\\
\hline\cline{1-0}
$$ & $$ & $$ & $$ & $$ & $$ & $$\\
$\tilde{s_6}$ & $-\theta^{2}_{2} -\theta^{2}_{3}$ & $-\theta^{2}_{1} -\theta^{2}_{2}$ & $-$ & $-\theta^{2}_{1} -\theta^{2}_{3}$ & $-$ & $-\theta^{2}_{1} -\theta^{2}_{2}$\\

$$ & $$ & $$ & $$ & $$ & $$ & $$\\
\hline\cline{1-0}
$$ & $$ & $$ & $$ & $$ & $$ & $$\\
$\tilde{s_3}$ & $-\theta^{2}_{1} + \theta^{2}_{2}$ & $-\theta^{2}_{2} + \theta^{2}_{3}$ & $-$ & $-$ & $-\theta^{2}_{1} -\theta^{2}_{3}$ & $-\theta^{2}_{1} + \theta^{2}_{2}$\\

$$ & $$ & $$ & $$ & $$ & $$ & $$\\
\hline\cline{1-0}
$$ & $$ & $$ & $$ & $$ & $$ & $$\\
$\tilde{s_4}$ & $-$ & $+\theta^{2}_{1} -\theta^{2}_{2}$ & $-\theta^{2}_{2} -\theta^{2}_{3}$ & $+\theta^{2}_{1}-\theta^{2}_{2}$ & $-$ & $+\theta^{2}_{1} -\theta^{2}_{3}$\\

$$ & $$ & $$ & $$ & $$ & $$ & $$\\
\hline\cline{1-0}
$$ & $$ & $$ & $$ & $$ & $$ & $$\\
$\tilde{s_5}$ & $+\theta^{2}_{1} -\theta^{2}_{2}$ & $-$ & $-\theta^{2}_{2} + \theta^{2}_{3}$ & $-$ & $+\theta^{2}_{1} - \theta^{2}_{2}$ & $+\theta^{2}_{1} + \theta^{2}_{3}$\\

$$ & $$ & $$ & $$ & $$ & $$ & $$\\
\hline\cline{1-0}
$$ & $$ & $$ & $$ & $$ & $$ & $$\\
$\tilde{s_1}$ & $-$ & $+\theta^{2}_{2} + \theta^{2}_{3}$ & $-\theta^{2}_{1} + \theta^{2}_{2}$ & $-\theta^{2}_{1} + \theta^{2}_{2}$ & $-\theta^{2}_{1} + \theta^{2}_{3}$ & $-$\\

$$ & $$ & $$ & $$ & $$ & $$ & $$\\
\hline\cline{1-0}
$$ & $$ & $$ & $$ & $$ & $$ & $$\\
$\tilde{s_2}$ & $+\theta^{2}_{2} + \theta^{2}_{3}$ & $-$ & $+\theta^{2}_{1} + \theta^{2}_{2}$ & $+\theta^{2}_{1} +\theta^{2}_{3}$ & $+\theta^{2}_{1} + \theta^{2}_{2}$ & $-$\\

$$ & $$ & $$ & $$ & $$ & $$ & $$\\
\hline
\end{tabular}
\
}
\eq
This table allows not only find the indices of singular points on minor surfaces, but also describe the direction (in terms of the local coordinates) in which the incoming/outgoing trajectories flow. In many cases it is enough to show, that two points are connected by a trajectory. For instance, we know, that there is only one outgoing trajectory from $\tilde{s_4}$, directed along $\theta_1$ (see the table \eqref{t2}). However, there are two potential endpoints of this trajectory, $\tilde{s_3}$ or $\tilde{s_6}$ and we cannot choose one of them without the help of the minor surfaces. But as on the other hand this trajectory goes along the surface $\psi_{31}$, where there is no singular points between $\tilde{s_4}$ and $\tilde{s_6}$, we conclude that this unique outgoing trajectory of $\tilde{s_4}$ will tend to $\tilde{s_6}$ when $t\to+\infty$.

Unfortunately, this table is not enough to describe all the trajectories. In order to make the picture complete we should consider the intersections of minor surfaces. It is easy to describe them here: for instance, the intersection $\psi_{12}\cap\psi_{13}$ (in $SO_3(\mathbb R)$) consists of special orthogonal $3\times 3$ matrices with zeros in the second and third places of the first row. Thus the second and the third columns of this matrix give a orthogonal basis of the subspace $\mathbb R^2\subset \mathbb R^3$, spanned by the second and the third coordinate vectors. The remaining (first) column should be equal to the unit vector, orthogonal to this subspace. So we conclude, that the intersection of these two minor surfaces consists of the matrices
$$
\begin{pmatrix}
1 & 0       & 0       \\
0 & \cos{t} & -\sin{t}\\
0 & \sin{t} &  \cos{t}
\end{pmatrix}\ \mbox{or}\ \begin{pmatrix}
                          -1 & 0        & 0       \\
                           0 & -\cos{t} & \sin{t} \\
                           0 & \sin{t} &  \cos{t}
\end{pmatrix}
$$
for some $t\in[0,\,2\pi]$. Now passing to factor space modulo the action of diagonal matrices, we can fix the element in the upper left corner equal to $1$ and factorize the remaining part by the group $\mathbb Z/2\mathbb Z$. This gives us a topological circle inside the factor space. Similarly if we take the intersection $\psi_{12}\cap\psi_{31}$, we shall obtain a wedge of two circles (in particular, this means that the intersection is no more a submanifold). Indeed, consider an orthogonal matrix which satisfies the conditions $\psi_{12}=0, \ \psi_{31}=0$:
$$
\begin{pmatrix}
\psi_{11} & 0 & \psi_{13}\\
\psi_{21} & \psi_{22} & \psi_{23}\\
0 & \psi_{32} &  \psi_{33}
\end{pmatrix}.\
$$
This matrix defines an orthogonal basis in $\mathbb R^3$: $f_{1}=\{ \psi_{11}, \ \psi_{21}, \ 0 \}, \ f_{2}=\{ 0, \ \psi_{22}, \ \psi_{32} \}, \ f_{3}=\{ \psi_{13}, \ \psi_{23}, \ \psi_{33} \}$. So the scalar product $(f_{1}, \ f_{2})$ should vanish and we see, that there can be two distinct types of such matrices:
$$
\begin{pmatrix}
\psi_{11} & 0 & \psi_{13}\\
0 & \psi_{22} & \psi_{23}\\
0 & \psi_{32} &  \psi_{33}
\end{pmatrix}\  \mbox{or}\ \begin{pmatrix}
\psi_{11} & 0 & \psi_{13}\\
\psi_{21} & 0 & \psi_{23}\\
0 & \psi_{32} &  \psi_{33}
\end{pmatrix}\
$$
And since $(f_{i}, \ f_{i}) =1$ we get:
$$
\begin{pmatrix}
1 & 0       & 0       \\
0 & \cos{t_{1}} & -\sin{t_{1}}\\
0 & \sin{t_{1}} &  \cos{t_{1}}
\end{pmatrix}\ \mbox{or}\ \begin{pmatrix}
                          -1 & 0        & 0       \\
                           0 & -\cos{t_{1}} & \sin{t_{1}} \\
                           0 & \sin{t_{1}} &  \cos{t_{1}}
\end{pmatrix}
$$
in the first case and
$$
\begin{pmatrix}
\cos{t_{2}} & 0 & -\sin{t_{2}}\\
\sin{t_{2}} & 0 & \cos{t_{2}}\\
0 & -1 &  0
\end{pmatrix}\ \mbox{or}\ \begin{pmatrix}
                          -\cos{t_{2}} & 0 & \sin{t_{2}}\\
                           \sin{t_{2}} & 0 & \cos{t_{2}}\\
                                     0 & 1 & 0
\end{pmatrix},
$$
in the second (here $t_{1}, \ t_{2} \in [0,\,2\pi]$). Each of these families determines a 1-dimensional submanifold in $SO_3(\mathbb R)$ which is isomorphic to 1-dimensional group $SO_2(\mathbb R)\cong S^1$ inside $\psi_{12}\bigcap\psi_{31}$. Thus in $Fl_3(\mathbb R)$, where we can fix $\pm1$ to $1$, $\psi_{12}\bigcap\psi_{31}$ is homeomorphic to the wedge of two circles.

It is easy to see that different critical points inside the double intersections might belong to different circles. For instance, in the example we considered above the first circle contains $\tilde{s}_{1}, \mbox{at} \ t_{1}= 0,\pi$ and $\tilde{s}_{4}, \mbox{at} \ t_{1}= \pm\pi/2$, and the second circle contains $\tilde{s}_{4}, \mbox{at} \ t_{2}= 0,2\pi$ and $\tilde{s}_{6}, \mbox{at} \ t_{2}=\pm\pi/2$.

In order to describe the phase transition from one point $\tilde{s}_{i}$ to another point $\tilde{s}_{j}$ inside the intersections of minor surfaces more visually, let us consider the following graph language. First of all recall that the infinitesimal behaviour of the gradient field of a Morse function at a singular point on a $2$-dimensional surface (and hence on the minor surfaces we consider here) falls into one of the following three classes according to the index of the Hessian (see (\ref{t3})):
\begin{itemize}
 \item index is equal to $0$ (both coefficients at the squares of the Morse coordinates are $+1$),
 \item index is equal to $1$ (one coefficient is positive, another coefficient is negative),
 \item index is equal to $2$ (both coefficients are negative).
\end{itemize}
In the case under consideration, Morse coordinates can be chosen proportional to $(\theta_1,\theta_2,\theta_3)$ (up to infinitesimal terms), and hence all the the signs are determined by the signs of the coefficients at $\theta_i^2$. We shall denote the corresponding vector fields symbolically by graphs with one vertex (singular point) and oriented edges, labeled with the coordinates $\theta_i$. The edge will be oriented towards the vertex, if the corresponding coefficient is negative, and from the vertex in the other case. We shall also draw the lines, that leave the vertex, above this vertex and the incoming lines below the vertex. Thus we have the following three basic pictures

\begin{figure}[h]
\begin{center}
\includegraphics[width=350pt,height=200pt]{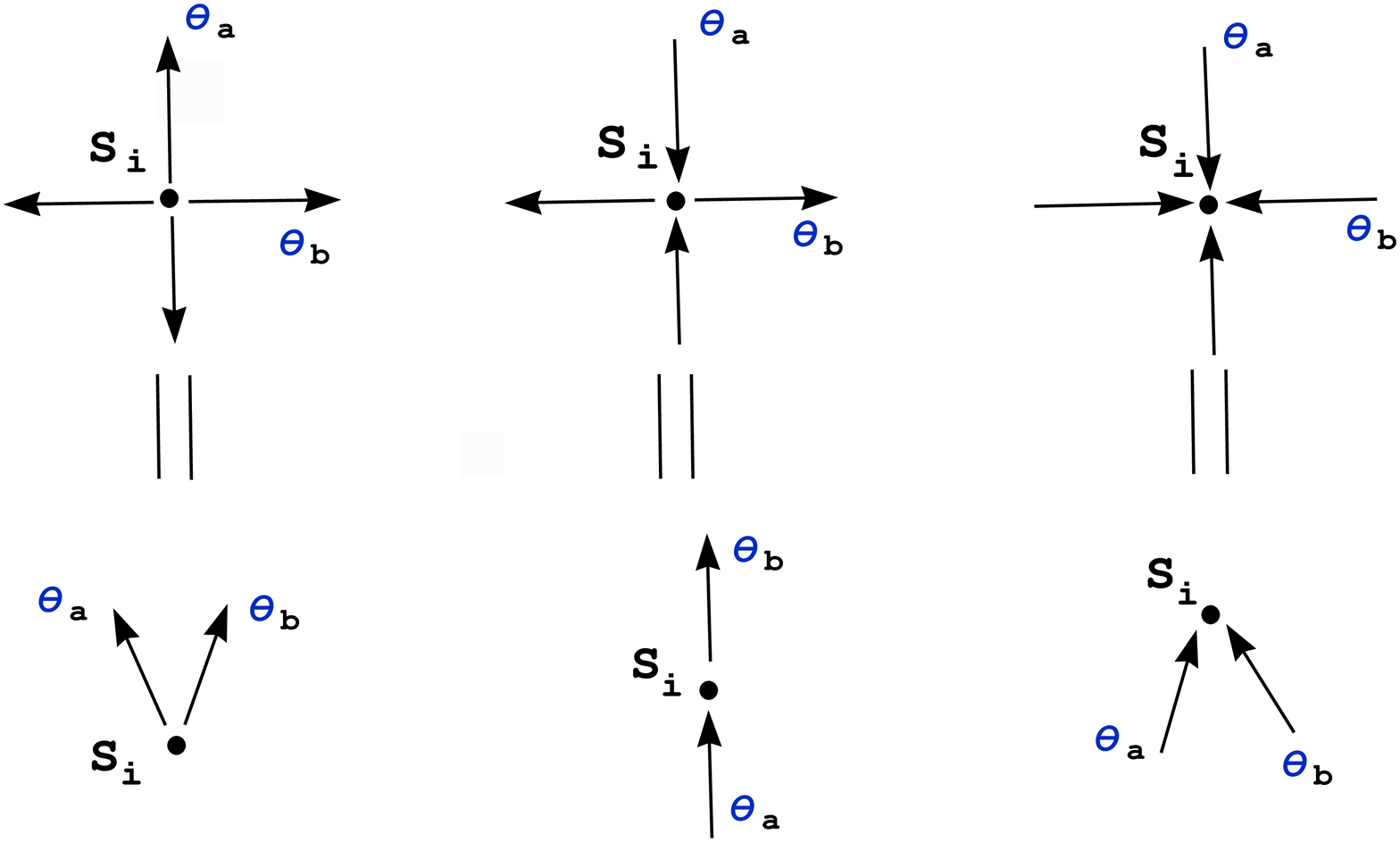}
\end{center}
\end{figure}
Using these building blocks we obtain the following graphical analogue of the table (\ref{t3}):
\begin{figure}[h]
\begin{center}
\includegraphics[width=350pt,height=350pt]{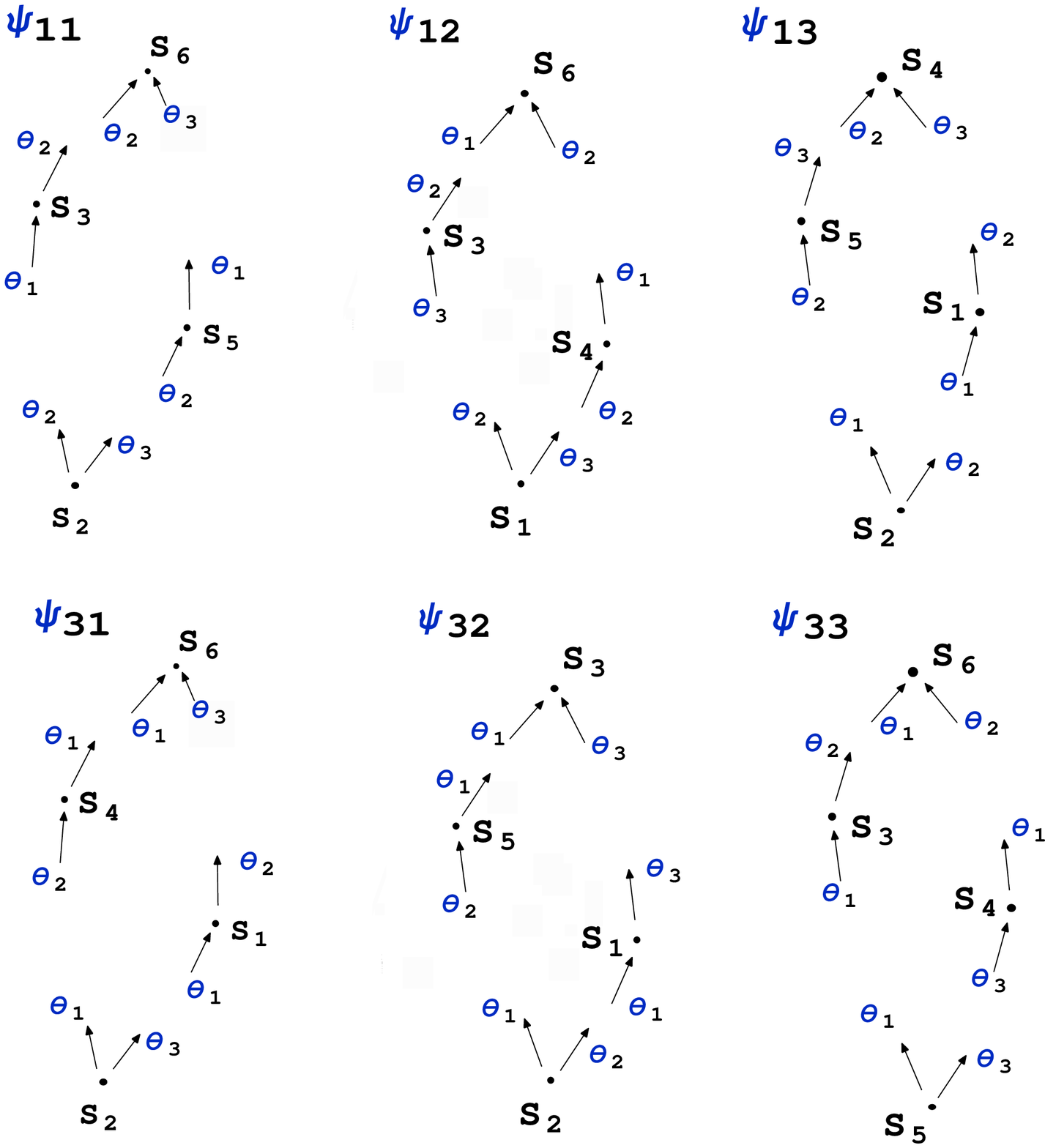}
\end{center}
\end{figure}

\newpage
One can use this graphical language to describe the transitions between the singular points inside the intersections of the minor surfaces. Take for instance $\psi_{12}\cap\psi_{13}$. From the pictures we see that this intersection contains only two points, $\tilde{s}_{4}$ and $\tilde{s}_{1}$. The next question is what trajectories are in the intersection of these two points. It follows from the graphical description of the surface $\psi_{13}=0$, that there is one trajectory in it entering $\tilde{s}_{1}$, and that this trajectory goes along the coordinate line of $\theta_{1}$. Since the value of the Morse function in $\tilde{s}_{1}$ is less than in $\tilde{s}_{4}$, this trajectory does not belong to the intersection. Similarly, the trajectory, that leaves $\tilde{s}_{1}$ along the coordinate $\theta_{3}$ (see the picture for the surface $\psi_{12}=0$) does not belong to the intersection, since it does not belong to the surface $\psi_{13}=0$. Analogously for $\tilde{s}_{4}$ we must delete the incoming trajectory along the coordinate  $\theta_{3}$ and the outcoming trajectory along the coordinate $\theta_{1}$.

The following table summarises these considerations, it describes all possible transitions from one point $\tilde{s}_{i}$ to another point $\tilde{s}_{j}$ which take place in the intersections of various minor surfaces $\psi_{ij} \cap \psi_{kl}$:

\beq{t2-2}
\tiny{
\begin{tabular}{|c|c|c|c|c|c|c|}
\hline\cline{1-0}
$$ & $$ & $$ & $$ & $$ & $$ & $$\\
$\ \ \ \tilde{S} \ \ \ $ & $\ \ \ \psi_{11} \ \ \ $ & $\ \ \ \psi_{12} \ \ \ $ & $\ \ \ \psi_{13} \ \ \ $ & $ \ \ \ \psi_{31} \ \ \ $ & $ \ \ \ \psi_{32} \ \ \ $ & $ \ \ \ \psi_{33} \ \ \ $\\

$$ & $$ & $$ & $$ & $$ & $$ & $$\\
\hline\cline{1-0}
$$ & $$ & $$ & $$ & $$ & $$ & $$\\
$\psi_{11}$ & $\times$ & $\tilde{S}_{3} \rightarrow \tilde{S}_{6}$ & $\tilde{S}_{2} \rightarrow \tilde{S}_{5}$ & $\tilde{S}_{2} \rightarrow \tilde{S}_{6}$ & $\tilde{S}_{2} \rightarrow \tilde{S}_{5}, \ \tilde{S}_{5} \rightarrow \tilde{S}_{3}$ & $\tilde{S}_{5} \rightarrow \tilde{S}_{3}, \ \tilde{S}_{3} \rightarrow \tilde{S}_{6}$\\

$$ & $$ & $$ & $$ & $$ & $$ & $$\\
\hline\cline{1-0}
$$ & $$ & $$ & $$ & $$ & $$ & $$\\
$\psi_{12}$ & $\times$ & $\times$ & $\tilde{S}_{1} \rightarrow \tilde{S}_{4}$ & $\tilde{S}_{1} \rightarrow \tilde{S}_{4}, \ \tilde{S}_{4} \rightarrow \tilde{S}_{6}$ & $\tilde{S}_{1} \rightarrow \tilde{S}_{3}$ & $\tilde{S}_{4} \rightarrow \tilde{S}_{6}, \ \tilde{S}_{3} \rightarrow \tilde{S}_{6}$\\

$$ & $$ & $$ & $$ & $$ & $$ & $$\\
\hline\cline{1-0}
$$ & $$ & $$ & $$ & $$ & $$ & $$\\
$\psi_{13}$ & $\times$ & $\times$ & $\times$ & $\tilde{S}_{2} \rightarrow \tilde{S}_{1}, \ \tilde{S}_{1} \rightarrow \tilde{S}_{4}$ & $\tilde{S}_{2} \rightarrow \tilde{S}_{1}, \ \tilde{S}_{2} \rightarrow \tilde{S}_{5}$ & $\tilde{S}_{5} \rightarrow \tilde{S}_{4}$\\

$$ & $$ & $$ & $$ & $$ & $$ & $$\\
\hline\cline{1-0}
$$ & $$ & $$ & $$ & $$ & $$ & $$\\
$\psi_{31}$ & $\times$ & $\times$ & $\times$ & $\times$ & $\tilde{S}_{2} \rightarrow \tilde{S}_{1}$ & $\tilde{S}_{4} \rightarrow \tilde{S}_{6}$\\

$$ & $$ & $$ & $$ & $$ & $$ & $$\\
\hline\cline{1-0}
$$ & $$ & $$ & $$ & $$ & $$ & $$\\
$\psi_{32}$ & $\times$ & $\times$ & $\times$ & $\times$ & $\times$ & $\tilde{S}_{5} \rightarrow \tilde{S}_{3}$\\

$$ & $$ & $$ & $$ & $$ & $$ & $$\\
\hline\cline{1-0}
$$ & $$ & $$ & $$ & $$ & $$ & $$\\
$\psi_{33}$ & $\times$ & $\times$ & $\times$ & $\times$ & $\times$& $\times$\\

$$ & $$ & $$ & $$ & $$ & $$ & $$\\
\hline
\end{tabular}
}
\eq
Finally observe, that however complicated the geometry of minor surfaces can be, they are $2$\/-dimensional surfaces at their generic points. Thus when we restrict our attention to a surface of this sort, we see that a generic trajectory on it should flow from the lowest to the highest singular point in this surface, and all other trajectories span a nowhere dense set. Similarly, a generic trajectory on $SO_3(\mathbb R)/N$ should flow from the minimum to the maximum value of $F_3$. Summing up, we conclude that the trajectories connect the singular points as it is described at the figure \ref{fig:fig4}. Here thin black arrows denote the "most singular"\ one-dimensional trajectories, "fat"\ blue arrows correspond to two-dimensional flows, "the fattest"\ red arrow is the generic trajectory on $SO_3(\mathbb R)/N=Fl_3(\mathbb R)$, i.e. the trajectories of these sort span locally a three-dimensional subspace.
\begin{figure}[tb]
\centering
		\includegraphics[scale=.55]{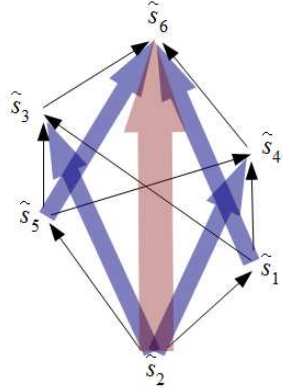}
	\caption{The trajectories in $SO_3(\mathbb R)$}
	\label{fig:fig4}
\end{figure}

As one can readily see, the structure of the arrows at this diagramm reproduces the Hasse diagramm of the group $S_3$, the Weyl group of $SL_3(\mathbb R)$. Moreover, the black arrows correspond to the minimal edges of the Hasse diagram, the blue arrows connect the points separated by one intermediate point in this diagram, and the red arrow connects the points, between which there are two other points.

\subsection{Trajectories on $SO_4(\mathbb R)$}
We now describe the trajectories of Toda system on $SO_4(\mathbb R)$ (or rather on $SO_4(\mathbb R)/SO_4(\mathbb R)\bigcap B^+_n\cong Fl_4(\mathbb R)$). The ideas we use here are the same as before, but the quantity of calculations is much bigger (in particular, there are $24$ singular points, which can be identified with the elements of $S_4$). So, we begin with fixing the order of eigenvalues:
$$
\lambda_{2} > \lambda_{1} > \lambda_{3} > 0>\lambda_{4}.
$$
The Morse function $F_4$ is equal to $F_{4}(a)=Tr(LN)$, where
$$
N=\begin{pmatrix} 0 & 0 & 0& 0\\ 0 & 1 & 0 & 0\\ 0 & 0 & 2 & 0\\ 0 & 0 & 0 & 3\end{pmatrix}.
$$
Once again one can show that the singular points of this function on $SO_4(\mathbb R)$ coincides with the group of $24\cdot 8=192$ matrices with $\mathrm{det}=1$, representing permutations of the vectors $\pm e_1,\,\pm e_2,\,\pm e_3,\,\pm e_4$. This group contains normal subgroup $SO_4(\mathbb R)\bigcap B^+_n$ of diagonal matrices, so that the phase space of our system can be identified with $Fl_4(\mathbb R)$, and the singular points of the induced system consists of $24$ points $\tilde{s}_{i_1i_2i_3i_4}$, where $i_1i_2i_3i_4$ is a permutation of four elements listed below (in $SO_4(\mathbb R),\ \ast$ should be replaced with $\pm1$, so that the determinant be equal to $1$):
\begin{equation*}
\footnotesize{
\begin{array}{c}
\tilde{s}_{1234}=
\left(
\begin{array}{cccc}
 \ast & 0 & 0 & 0 \\
 0 & \ast & 0 & 0 \\
 0 & 0 & \ast & 0 \\
 0 & 0 & 0 & \ast
\end{array}
\right), \ \ \
\tilde{s}_{2341}=
\left(
\begin{array}{cccc}
 0 & \ast & 0 & 0 \\
 0 & 0 & \ast & 0 \\
 0 & 0 & 0 & \ast \\
 \ast & 0 & 0 & 0
\end{array}
\right),\\
\tilde{s}_{3412}=
\left(
\begin{array}{cccc}
 0 & 0 & \ast & 0 \\
 0 & 0 & 0 & \ast \\
 \ast & 0 & 0 & 0 \\
 0 & \ast & 0 & 0
\end{array}
\right), \ \ \
\tilde{s}_{4123}=
\left(
\begin{array}{cccc}
 0 & 0 & 0 & \ast \\
 \ast & 0 & 0 & 0 \\
 0 & \ast & 0 & 0 \\
 0 & 0 & \ast & 0
\end{array}
\right),
\end{array}
}
\end{equation*}
\begin{equation*}
\footnotesize{
\begin{array}{c}
\tilde{s}_{1324}=
\left(
\begin{array}{cccc}
 \ast & 0 & 0 & 0 \\
 0 & 0 & \ast & 0 \\
 0 & \ast & 0 & 0 \\
 0 & 0 & 0 & \ast
\end{array}
\right), \ \ \
\tilde{s}_{3241}=
\left(
\begin{array}{cccc}
 0 & 0 & \ast & 0 \\
 0 & \ast & 0 & 0 \\
 0 & 0 & 0 & \ast \\
 \ast & 0 & 0 & 0
\end{array}
\right), \\
\tilde{s}_{2413}=
\left(
\begin{array}{cccc}
 0 & \ast & 0 & 0 \\
 0 & 0 & 0 & \ast \\
 \ast & 0 & 0 & 0 \\
 0 & 0 & \ast & 0
\end{array}
\right), \ \ \
\tilde{s}_{4132}=
\left(
\begin{array}{cccc}
 0 & 0 & 0 & \ast \\
 \ast & 0 & 0 & 0 \\
 0 & 0 & \ast & 0 \\
 0 & \ast & 0 & 0
\end{array}
\right),
\end{array}
}
\end{equation*}
\begin{equation*}
\footnotesize{
\begin{array}{c}
\tilde{s}_{1342}=
\left(
\begin{array}{cccc}
 \ast & 0 & 0 & 0 \\
 0 & 0 & \ast & 0 \\
 0 & 0 & 0 & \ast \\
 0 & \ast & 0 & 0
\end{array}
\right), \ \ \
\tilde{s}_{3421}=
\left(
\begin{array}{cccc}
 0 & 0 & \ast & 0 \\
 0 & 0 & 0 & \ast \\
 0 & \ast & 0 & 0 \\
 \ast & 0 & 0 & 0
\end{array}
\right), \\
\tilde{s}_{4213}=
\left(
\begin{array}{cccc}
 0 & 0 & 0 & \ast \\
 0 & \ast & 0 & 0 \\
 \ast & 0 & 0 & 0 \\
 0 & 0 & \ast & 0
\end{array}
\right), \ \ \
\tilde{s}_{2134}=
\left(
\begin{array}{cccc}
 0 & \ast & 0 & 0 \\
 \ast & 0 & 0 & 0 \\
 0 & 0 & \ast & 0 \\
 0 & 0 & 0 & \ast
\end{array}
\right),
\end{array}
}
\end{equation*}
\begin{equation*}
\footnotesize{
\begin{array}{c}
\tilde{s}_{1432}=
\left(
\begin{array}{cccc}
 \ast & 0 & 0 & 0 \\
 0 & 0 & 0 & \ast \\
 0 & 0 & \ast & 0 \\
 0 & \ast & 0 & 0
\end{array}
\right), \ \ \
\tilde{s}_{4321}=
\left(
\begin{array}{cccc}
 0 & 0 & 0 & \ast \\
 0 & 0 & \ast & 0 \\
 0 & \ast & 0 & 0 \\
 \ast & 0 & 0 & 0
\end{array}
\right), \\
\tilde{s}_{3214}=
\left(
\begin{array}{cccc}
 0 & 0 & \ast & 0 \\
 0 & \ast & 0 & 0 \\
 \ast & 0 & 0 & 0 \\
 0 & 0 & 0 & \ast
\end{array}
\right), \ \ \
\tilde{s}_{2143}=
\left(
\begin{array}{cccc}
 0 & \ast & 0 & 0 \\
 \ast & 0 & 0 & 0 \\
 0 & 0 & 0 & \ast \\
 0 & 0 & \ast & 0
\end{array}
\right),
\end{array}
}
\end{equation*}
\begin{equation*}
\footnotesize{
\begin{array}{c}
\tilde{s}_{1423}=
\left(
\begin{array}{cccc}
 \ast & 0 & 0 & 0 \\
 0 & 0 & 0 & \ast \\
 0 & \ast & 0 & 0 \\
 0 & 0 & \ast & 0
\end{array}
\right), \ \ \
\tilde{s}_{4231}=
\left(
\begin{array}{cccc}
 0 & 0 & 0 & \ast \\
 0 & \ast & 0 & 0 \\
 0 & 0 & \ast & 0 \\
 \ast & 0 & 0 & 0
\end{array}
\right), \\
\tilde{s}_{2314}=
\left(
\begin{array}{cccc}
 0 & \ast & 0 & 0 \\
 0 & 0 & \ast & 0 \\
 \ast & 0 & 0 & 0 \\
 0 & 0 & 0 & \ast
\end{array}
\right), \ \ \
\tilde{s}_{3142}=
\left(
\begin{array}{cccc}
 0 & 0 & \ast & 0 \\
 \ast & 0 & 0 & 0 \\
 0 & 0 & 0 & \ast \\
 0 & \ast & 0 & 0
\end{array}
\right),
\end{array}
}
\end{equation*}
\begin{equation*}
\footnotesize{
\begin{array}{c}
\tilde{s}_{1243}=
\left(
\begin{array}{cccc}
 \ast & 0 & 0 & 0 \\
 0 & \ast & 0 & 0 \\
 0 & 0 & 0 & \ast \\
 0 & 0 & \ast & 0
\end{array}
\right), \ \ \
\tilde{s}_{2431}=
\left(
\begin{array}{cccc}
 0 & \ast & 0 & 0 \\
 0 & 0 & 0 & \ast \\
 0 & 0 & \ast & 0 \\
 \ast & 0 & 0 & 0
\end{array}
\right), \\
\tilde{s}_{4312}=
\left(
\begin{array}{cccc}
 0 & 0 & 0 & \ast \\
 0 & 0 & \ast & 0 \\
 \ast & 0 & 0 & 0 \\
 0 & \ast & 0 & 0
\end{array}
\right), \ \ \
\tilde{s}_{3124}=
\left(
\begin{array}{cccc}
 0 & 0 & \ast & 0 \\
 \ast & 0 & 0 & 0 \\
 0 & \ast & 0 & 0 \\
 0 & 0 & 0 & \ast
\end{array}
\right).
\end{array}
}
\end{equation*}
As before, we introduce the local coordinates at a point $g\in SO_4(\mathbb R)$, pulling them from the Lie algebra $\mathfrak{so}_4$, whose generic element $\Theta$ is equal to
\beq{tN4}
\Theta = \left(
\begin{array}{cccc}
 0 & \theta _1 & \theta _3 & \theta _4 \\
 -\theta _1 & 0 & \theta _2 & \theta _5 \\
 -\theta _3 & -\theta _2 & 0 & \theta _6 \\
 -\theta _4 & -\theta _5 & -\theta _6 & 0
\end{array}
\right).\\
\eq

\begin{figure}[tb]
\begin{center}
\includegraphics[width=350pt,height=200pt]{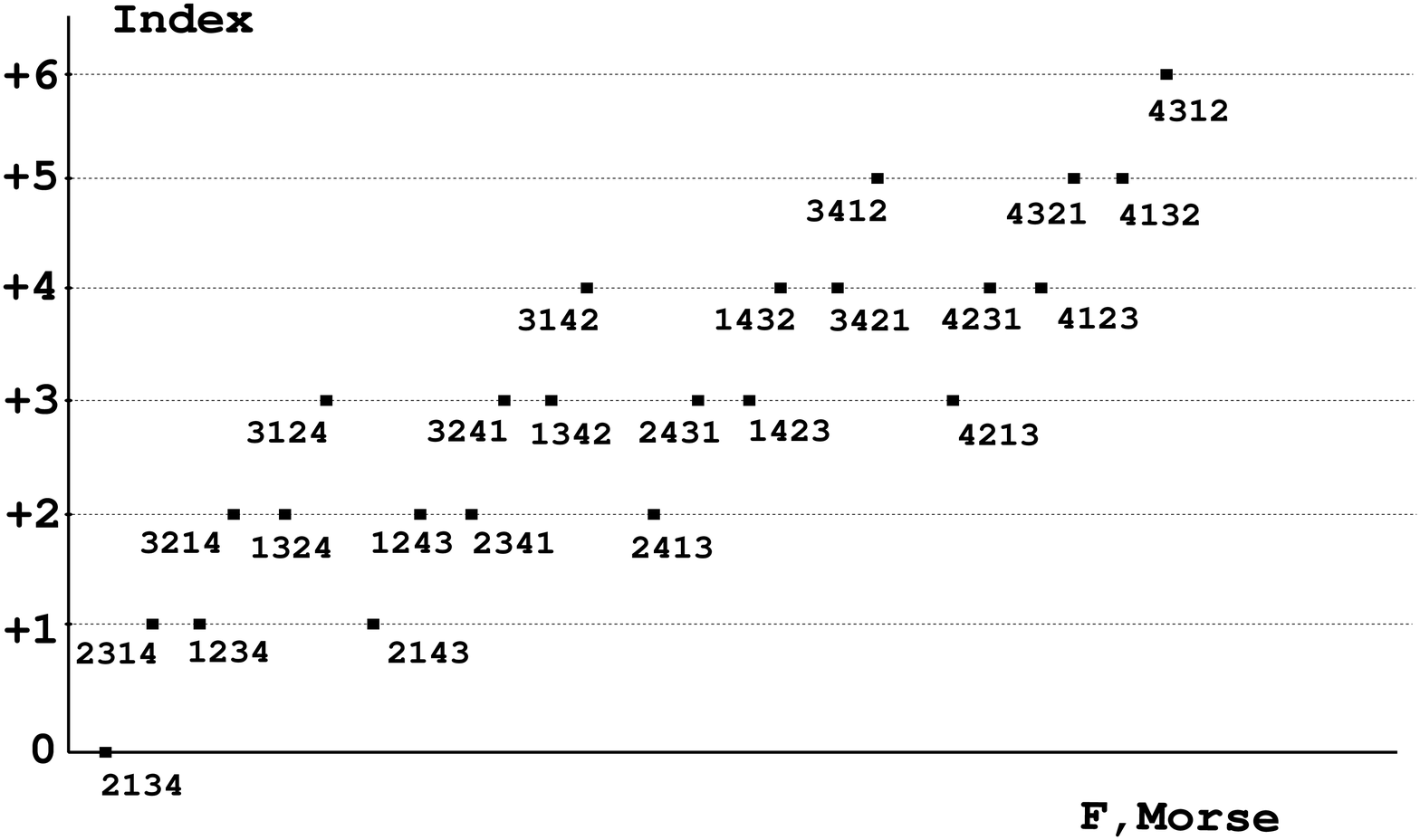}
\caption{Singular points and their indices in $SO_4(\mathbb R)$}
\label{fig:distr}
\end{center}
\end{figure}
We can now express the second order approximation of $F_4$ at a point $\tilde s_{i_1i_2i_3i_4}$ in the terms of coordinates $(\theta_1,\dots,\theta_6)$, for instance:
$$
\begin{aligned}
dF_{4}(\tilde s_{1234})&=\theta _1^2 \left(\lambda _1-\lambda _2\right)+2 \theta _3^2 \left(\lambda _1-\lambda _3\right)+\theta _2^2\\ &\quad+\left(\lambda _2-\lambda _3\right)+3 \theta _4^2 \left(\lambda _1-\lambda _4\right)+2 \theta _5^2
   \left(\lambda _2-\lambda _4\right)+\theta _6^2 \left(\lambda _3-\lambda _4\right).
\end{aligned}
$$

To derive the corresponding approximations for the rest of $\tilde s_{i_1i_2i_3i_4}$, one should permute the eigenvalues $\lambda_i$ by $i_1i_2i_3i_4$. Thus we get the following picture of the singular values of $F_4$ and their indices (see figure \ref{fig:distr}). At this picture we have shown all $24$ singular points, vertical axis shows the indices of these points and the horizontal axis measures the corresponding values of the function $F_4$.

\begin{figure}[tb]
\begin{center}
\includegraphics[width=300pt,height=320pt]{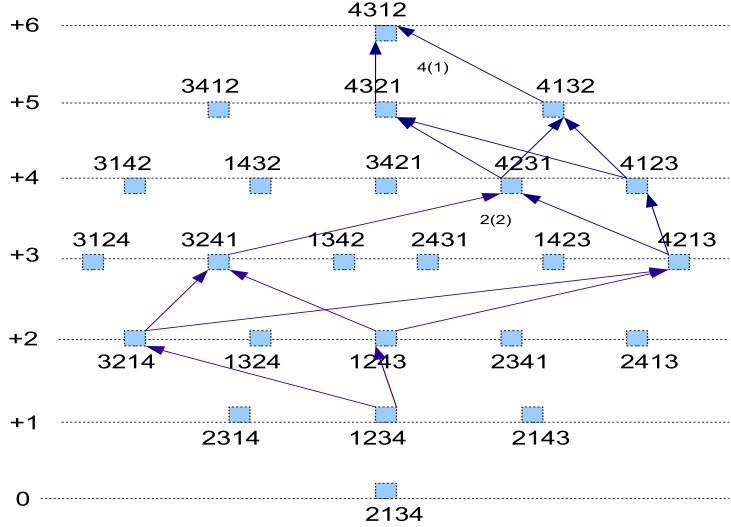}

\vspace{-2cm}
\caption{Examples of $SO_3(\mathbb R)$ clusters inside $SO_4(\mathbb R)$}
\label{fig:clustex}
\end{center}
\end{figure}
As one can see, there are four clusters of six elements each, the positions of points at every cluster resembling very much the points in the phase portrait of $SO_3(\mathbb R)$. These clusters do, actually, correspond to the submanifolds within $SO_4(\mathbb R)$ isomorphic to $SO_3(\mathbb R)$; the lower left cluster here is a genuine subgroup inside  $SO_4(\mathbb R)$, isomorphic to $SO_3(\mathbb R)$: the subgroup, consisting of orthogonal transforms which leave invariant the fourth basis vector. There are four such subgroups (one for every vector $e_i,\ i=1,\dots,4$). We  shall denote these groups by $SO^i_3(\mathbb R),\ i=1,\dots,4$. It is straightforward to see, that these subgroups are invariant with respect to the dynamic system under consideration. For instance, one can prove this by the direct calculation with formula \eqref{Lax-GenToda}. Moreover, this calculation shows, that the dynamical system induced on these submanifolds from the symmetric Toda systtem on $SO_4(\mathbb R)$ coincides with the full symmetric Toda system on $SO_3(\mathbb R)$, if we use the natural identifications (just observe, that the condition $\sum\lambda_i=0$ is not necessary for the definition of the vector field $A$, see equation \eqref{Lax-GenToda}). Thus the restriction of the flow diagram to these subgroups coincides with the diagrams of Bruhat order for $SO_3(\mathbb R)$.

\begin{figure}[tb]
\begin{center}
\includegraphics[width=300pt,height=320pt]{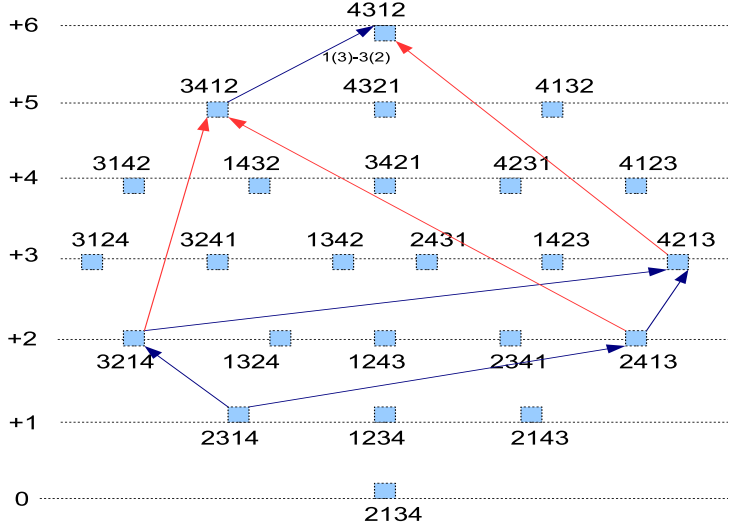}

\vspace{-2cm}
\caption{An exotic $SO_3(\mathbb R)$ cluster inside $SO_4(\mathbb R)$}
\label{fig:clustex2}
\end{center}
\end{figure}
This observation can be further generalized, yielding the full description of the phase flow on $SO_4(\mathbb R)$. Let $s$ be one of the matrices $\tilde s_{i_1i_2i_3i_4}$ listed above. Suppose that the permutation corresponding to it sends $1$ to $i_1$. We can assume, that the ${1i_1}$\/-th matrix element of $s$ is equal to $1$ and not to $-1$. Then consider the coset space
$$
SO^{1i_1}_3=SO^1_3(\mathbb R) \cdot s.
$$
This set consists of all the orthogonal linear operators, which send $e_1$ to $e_{i_1}$. It is evidently diffeomorphic to the group $SO_3(\mathbb R)$, although it does not have a natural group structure.

Let now $\Psi\in SO^{1i_1}_3$ be an arbitrary element. As it follows from the equation \eqref{Lax-GenToda}, the vector field $A=A(\Psi,\Lambda)$ at this point is equal to the right translation by $\tilde s_{i_1i_2i_3i_4}$ of the field
$A(\Psi\tilde s^{-1},\,\tilde s\Lambda\tilde s^{-1})$ on $SO^1_3(\mathbb R)$ which can be identified with the vector field $A$ on $SO_3(\mathbb R)$, provided we take the eigenvalues equal to a permutation of $\lambda_2,\,\lambda_3,\,\lambda_4$. Thus the phase structure of the restriction of the Toda system to this subspace is equal to the translation of the corresponding phase portrait from $SO_3(\mathbb R)$ via right multiplication by $\tilde s_{i_1i_2i_3i_4}$.

The same construction is applicable if we replace $1$ with $2,\,3$ or $4$. Thus we get a collection of $16$ intersecting clusters inside the set of all critical points of the Toda system for $SO_4(\mathbb R)$, such that the restrictions of phase portrait to any of these clusters coincide with the portrait of $SO_3(\mathbb R)$ (after suitable identifications). In table \eqref{t1} below we give a list of these clusters. And figure \ref{fig:clustex} shows the corresponding phase portraits: one, corresponding to $SO_3^{14}$ and the other given by $SO_3^{22}=SO_3^2(\mathbb R)$. Observe however, that the indices of singular points inside a subspace $SO_3^{ij}$ and inside $SO_4(\mathbb R)$ need not coinside. For example, at figure \ref{fig:clustex2} we have drawn the scheme of the cluster $SO_3^{31}$. As one can see, three of the arrows, which correspond to the singular points whose Morse index at the level of $SO_3(\mathbb R)$ differ by $1$, at this picture connect points, whose Morse indices in $SO_4(\mathbb R)$ differ by $2$. Moreover, the same pair of singular points can represent neighbors in one cluster and non-neighbors in another.

Summing this up, we get the following table showing all the $SO_3(\mathbb R)$ subspaces inside $SO_4(\mathbb R)$: here we use the notation introduced above to describe the subspace we deal with; further, in the second column we give the list of singular points, that fall in this subspace; finally, in the last column we list several minors conditions, that single out the subspace.

As one can show by direct inspection, all the eigenspaces of Hessians of the singular points in $SO_4(\mathbb R)$, such that their dimension is less than $4$, are covered by one or another of $SO_3^{ij}$ subspaces listed in table \eqref{t1}. Thus we only need to check the configurations of minor surfaces with codimensions $2$ or less. This is done by a direct inspection. Thus we obtain the following picture of the flows inside $SO_4(\mathbb R)$ (or rather inside the corresponding flag space $Fl_4(\mathbb R)$), see figure \ref{fig:bruhat1}. Here we show only the one-dimensional flows on $SO_4(\mathbb R)$, that is the flows, corresponding to $1$\/-dimensional eigen-subspaces of Hessians at singular points (for references we have added the cyclic decompositions of the corresponding permutations).

\beq{t1}
\tiny{
\begin{tabular}{|c|c|c|}

\hline\cline{1-0}
$$ & $$ & $$\\
\footnotesize{\mbox{Cluster}} & \footnotesize{\mbox{Points}} & \footnotesize{\mbox{Minors}} \\

$$ & $$ & $$\\
\hline\cline{1-0}
$$ & $$ & $$\\
\footnotesize{$SO_3^{41}$} & \footnotesize{$4312, 4132, 4321, 4123, 4231, 4213$} & \footnotesize{$\psi_{11}, \psi_{12}, \psi_{13}, \psi_{44}$}\\

$$ & $$ & $$\\
\hline\cline{1-0}
$$ & $$ & $$\\
\footnotesize{$SO_3^{42}$} & \footnotesize{$3412, 3421, 1432, 1423, 2431, 2413$} & \footnotesize{$\psi_{14}, \psi_{44}, M_{\frac{12}{12}}, M_{\frac{12}{13}}, M_{\frac{12}{23}}$}\\

$$ & $$ & $$\\
\hline\cline{1-0}
$$ & $$ & $$\\
\footnotesize{$SO_3^{43}$} & \footnotesize{$3142, 1342, 3241, 2341, 1243, 2143$} & \footnotesize{$\psi_{14}, \psi_{44}, M_{\frac{12}{14}}, M_{\frac{12}{24}}, M_{\frac{12}{34}}$}\\

$$ & $$ & $$\\
\hline\cline{1-0}
$$ & $$ & $$\\
\footnotesize{$SO_3^{44}$} & \footnotesize{$3124, 1324, 3214, 1234, 2314, 2134$} & \footnotesize{$\psi_{14}, \psi_{41}, \psi_{42}, \psi_{43}$}\\

$$ & $$ & $$\\
\hline\cline{1-0}
$$ & $$ & $$\\
\footnotesize{$SO_3^{21}$} & \footnotesize{$2431, 2413, 2341, 2143, 2314, 2134$} & \footnotesize{$\psi_{11}, \psi_{13}, \psi_{14}, \psi_{42}$}\\

$$ & $$ & $$\\
\hline\cline{1-0}
$$ & $$ & $$\\
\footnotesize{$SO_3^{22}$} & \footnotesize{$4231, 4213, 3241, 1243, 3214, 1234$} & \footnotesize{$\psi_{12}, \psi_{42}, M_{\frac{12}{13}}, M_{\frac{12}{14}}, M_{\frac{12}{34}}$}\\

$$ & $$ & $$\\
\hline\cline{1-0}
$$ & $$ & $$\\
\footnotesize{$SO_3^{23}$} & \footnotesize{$4321, 4123, 3421, 1423, 3124, 1324$} & \footnotesize{$\psi_{12}, \psi_{42}, M_{\frac{12}{12}}, M_{\frac{12}{23}}, M_{\frac{12}{24}}$}\\

$$ & $$ & $$\\
\hline\cline{1-0}
$$ & $$ & $$\\
\footnotesize{$SO_3^{24}$} & \footnotesize{$4312, 4132, 3412, 1432, 3142, 1342$} & \footnotesize{$\psi_{12}, \psi_{41}, \psi_{43}, \psi_{44}$}\\

$$ & $$ & $$\\
\hline\cline{1-0}
$$ & $$ & $$\\
\footnotesize{$SO_3^{31}$} & \footnotesize{$3412, 3421, 3142, 3241, 3124, 3214$} & \footnotesize{$\psi_{11}, \psi_{12}, \psi_{14}, \psi_{43}$}\\

$$ & $$ & $$\\
\hline\cline{1-0}
$$ & $$ & $$\\
\footnotesize{$SO_3^{32}$} & \footnotesize{$4312, 4321, 1342, 2341, 1324, 2314$} & \footnotesize{$\psi_{13}, \psi_{43}, M_{\frac{12}{12}}, M_{\frac{12}{14}}, M_{\frac{12}{24}}$}\\

$$ & $$ & $$\\
\hline\cline{1-0}
$$ & $$ & $$\\
\footnotesize{$SO_3^{33}$} & \footnotesize{$4132, 4231, 1432, 2431, 1234, 2134$} & \footnotesize{$\psi_{13}, \psi_{43}, M_{\frac{12}{13}}, M_{\frac{12}{23}}, M_{\frac{12}{34}}$}\\

$$ & $$ & $$\\
\hline\cline{1-0}
$$ & $$ & $$\\
\footnotesize{$SO_3^{34}$} & \footnotesize{$4123, 4213, 1423, 2413, 1243, 2143$} & \footnotesize{$\psi_{13}, \psi_{41}, \psi_{42}, \psi_{44}$}\\

$$ & $$ & $$\\
\hline\cline{1-0}
$$ & $$ & $$\\
\footnotesize{$SO_3^{11}$} & \footnotesize{$1432, 1423, 1342, 1243, 1324, 1234$} & \footnotesize{$\psi_{12}, \psi_{13}, \psi_{14}, \psi_{41}$}\\

$$ & $$ & $$\\
\hline\cline{1-0}
$$ & $$ & $$\\
\footnotesize{$SO_3^{12}$} & \footnotesize{$4132, 4123, 3142, 2143, 3124, 2134$} & \footnotesize{$\psi_{11}, \psi_{41}, M_{\frac{12}{23}}, M_{\frac{12}{24}}, M_{\frac{12}{34}}$}\\

$$ & $$ & $$\\
\hline\cline{1-0}
$$ & $$ & $$\\
\footnotesize{$SO_3^{13}$} & \footnotesize{$4312, 4213, 3412, 2413, 3214, 2314$} & \footnotesize{$\psi_{11}, \psi_{41}, M_{\frac{12}{12}}, M_{\frac{12}{13}}, M_{\frac{12}{14}}$}\\

$$ & $$ & $$\\
\hline\cline{1-0}
$$ & $$ & $$\\
\footnotesize{$SO_3^{14}$} & \footnotesize{$4321, 4231, 3421, 2431, 3241, 2341$} & \footnotesize{$\psi_{11}, \psi_{42}, \psi_{43}, \psi_{44}$}\\

$$ & $$ & $$\\
\hline
\end{tabular}
}
\eq

\newpage

\begin{figure}[!t]
\begin{center}
\includegraphics[scale=.75]{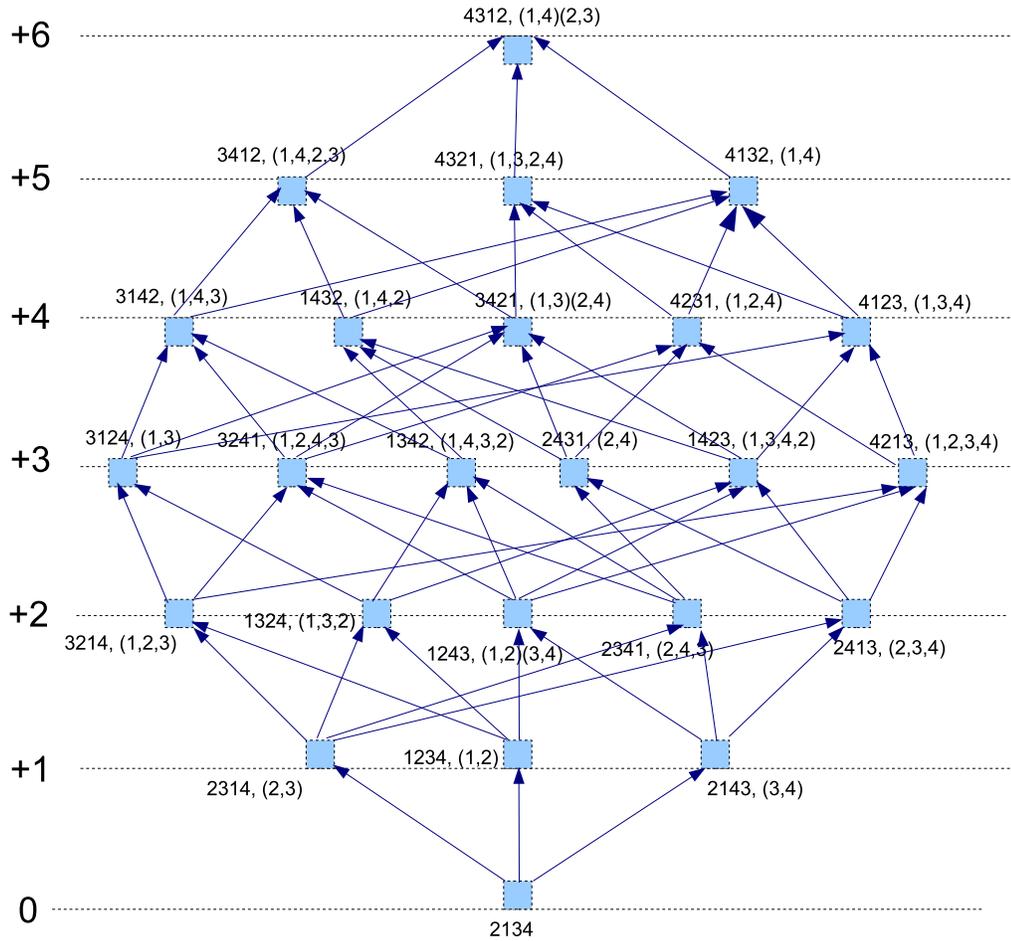}
\vspace{-2cm}
\caption{The $1$\/-dimensional flows at $SO_4(\mathbb R)$}
\label{fig:bruhat1}
\end{center}
\end{figure}

As one can see, it is just the Hasse diagramm of the Bruhat order on the symmetric group $S_4$, or rather its image under the action of the permutation $\begin{pmatrix}1 & 2 & 3 & 4\\ 2 & 1 & 3 & 4\end{pmatrix}$; this is the permutation, which reorders the eigenvalues $\lambda_i$ so that $\lambda_1<\lambda_2<\lambda_3<\lambda_4$. Thus we get for  $SO_4(\mathbb R)$ the following result, just as for $SO_3(\mathbb R)$ before:

\begin{prop}
Two singular points $x$ and $y$ of the full symmetric Toda flow on $SO_4(\mathbb R)$ are connected by a trajectory if and only if $x< y$ in Bruhat order of $S_4$ (after suitable permutation, corresponding to the chosen order of eigenvalues). Moreover, for any $n=2,3,4,5,6$ the  $n$\/-dimensional flows of the Toda system correspond to the length $n$ segments of Bruhat order (recall, that a segment in partial ordered set is a sub poset $[x;y]=\{z\in P|x< z<y\}$; a segment is said to have length $n$, if for any chain of the maximal length $x=x_0<x_1<\dots< x_k=y$ inside it $k=n$).
\end{prop}

\subsection{The general case}
Let us now prove that in the case of arbitrary $n$ the similar conclusions should hold, namely:

\begin{theorem}
\label{theomag} The singular points of full symmetric Toda flow on $SO_n(\mathbb R)$ can be identified with the elements of the symmetric group $S_n$ so that two points are connected by a trajectory if and only if they are comparable in Bruhat order on $S_n$ (after suitable permutation, corresponding to the chosen order of eigenvalues).
\end{theorem}

Moreover, we shall also prove that there is a $k$\/-parametric family of such trajectories, if one can insert $k$ points in Bruhat order between the given two (see corollary \ref{corMS}).

The remark, concerning the order of the eigenvalues means that if we identify $S_n$ with the group of permutations of the eigenvalues of $L$ so that the standard (default) order of the eigenvalues is $\lambda_1<\lambda_2<\dots<\lambda_n$, then the corresponding phase portrait can be identified with the Hasse diagramm of the Bruhat order. Otherwise, one should multiply the elements of $S_n$ in Hasse diagramm by the corresponding permutation. In previous sections we have described the low-dimensional case ($n=3,\,4$), where this statement can be proved by direct inspection. It is our purpose now to prove this statement for arbitrary $n$.

\subsubsection{Supporting remarks}
Before we give the full proof of the theorem \ref{theomag}, let us give few observations to support it. We begin with the following fact: if an element $\alpha\in S_n,\ n\ge3$ precedes $\beta$ with respect to Bruhat order, so that there is no other elements of $S_n$ between them, then there is an element $i\in\{1,\dots,n\}$ such that $\alpha(i)=\beta(i)$. This is an almost direct consequence of the definition of Bruhat order (see section \ref{sectorder}), since such neighbors always differ only by a single swap of two elements.

Now we reason by induction on $n$: suppose that for all $k\le n$ we know that any two elements in $S_k$, which are direct neighbors in Bruhat order (i.e. there are now other elements between them), are connected by a trajectory of Toda flow inside $SO_k(\mathbb R)/SO_k(\mathbb R)\bigcap B^+_k=Fl_k(\mathbb R)$.

Consider the singular points $\tilde s_i\in S_{n+1}$ of the system in dimension $n+1$. We shall always regard the elements of $S_{n+1}\subset FL_n(\mathbb R)$ as the classes of equivalences of matrices inside $SO_{n}(\mathbb R)$; similarly, we shall speak rather about the submanifolds in $SO_{n+1}(\mathbb R)$, and not about their images in $Fl_{n+1}(\mathbb R)$. Thus one can consider $(n+1)^2$ submanifolds $SO_n^{ij}\subseteq SO_{n+1}(\mathbb R)$ isomorphic to the group $SO_n(\mathbb R)$: they are defined in the same way we have previously defined the submanifolds $SO_3^{ij}\subset SO_4(\mathbb R)/SO_4(\mathbb R)\bigcap B^+_4$. In particular, there are $n+1$ subgroups $SO_n^i(\mathbb R)\subseteq SO_{n+1}(\mathbb R)$ isomorphic to $SO_n(\mathbb R)$. As before it is evident, that the restriction of the Toda system to these subgroups can be identified with the Toda system on $SO_n(\mathbb R)$, and hence the trajectories of the restriction of Toda system to any of $SO_n^i(\mathbb R)$ connect two points inside the corresponding part of $S_n$ if and only if they are comparable with respect to the Bruhat order in $S_n$. The same is true for the restrictions of this system to all the other submanifolds $SO_n^{ij}\subset SO_{n+1}(\mathbb R)$.

Thus we see that \textit{any two elements $\alpha,\,\beta\in S_{n+1}$ belong to one submanifold $SO_n^{ij}$ for some $i,j=1,\dots,n+1$ if and only if there exist an element $1\le i\le n+1$, such that $\alpha(i)=j=\beta(i)$}. Now the following is true by definition of Bruhat order: every two points $\alpha,\,\beta$, which are direct neighbours with respect to the Bruhat order in $S_{n+1}$, belong to the same $SO_n^{ij}$ for certain $i,j$. Since Bruhat order is hereditary in the sense that the restriction of this order from $S_{n+1}$ to $S_n\subseteq S_{n+1}$, which is embedded as the set of all permutations preserving certain element, coincides with the Bruhat order on $S_n$ (though some nondirect neighbours can become direct in smaller subgroup) we conclude from the observation above, that $\alpha$ and $\beta$ shall be connected by a trajectory of Toda flow on $SO_{n+1}(\mathbb R)$.

Thus we have proved, that \textit{any two neighbors in Bruhat order will be always connected by a Toda trajectory}. However, we still do not know what happens, if these two points are not neighbors in this order. And we cannot tell anything about the dimension of the space of all trajectories between two points.

One can try to use similar ideas to prove the reverse statement (i.e. to show, that the points that are not comparable with respect to the Bruhat order cannot be connected by a trajectory). Instead of this, let us proceed by proving the general statement.

\subsubsection{Proof of the theorem \ref{theomag}}
\begin{proof}
As we have mentioned, \textit{Toda system preserves varieties given by the equations $M_{\frac{1\dots k}{i_1\dots i_k}}=0$}, where $M_{\frac{1\dots k}{i_1\dots i_k}}$ is the determinant of the submatrix of $\Psi$ spanned by the intersections of the first $k$ rows and the columns $i_1,\dots,i_k$. Hence, if equations which determine the Bruhat cell are of this form, then this Bruhat cell is preserved by the Toda flow. In this case any point in such cell $B_\alpha$ can flow only to a point that lies in the closure of this cell, which consists of the cells corresponding to $\beta\in S_{n+1}$ such that $\beta<\alpha$ in Bruhat order.

However not all Schubert cells are defined by equations of the necessary form, since there can be rectangular submatrices of type $k\times l$ with $k<l$ (see section \ref{sectorder}). Thus we shall need the following proposition to be able to handle all Bruhat cells:
\begin{lem}
If the (upper-left) submatrix $\Psi^{i,j}(0)$ of $\Psi(0)$ has rank less than $\min(i,j)$, then the same is true for all $t$: $\mathrm{rk}\,\Psi^{i,j}(t)<\min(i,j)$.
\end{lem}
\begin{proof}
The proof is based on the following observations: first, to say, that $\mathrm{rk}\,\Psi^{i,j}(t)<\min(i,j)$ is equivalent to impose the condition that all the minors of the $\min(i,j)\times\min(i,j)$ submatrices in $\Psi(t)$ are equal to $0$. If $i=j$ or $i>j$ (i.e. the matrix is square or "horizontal"), then the conclusion of the theorem follows at once from the fact, that equations $M_{\frac{1\dots k}{i_1\dots i_k}}(t)=0$ hold identically, if they are verified for $t=0$. If $i<j$, we need to work a little more.

Namely consider the minors of $\Psi$, which are equal to the intersection of the first $k<n$ columns of $\Psi$ with an arbitrary set of $k$ rows. Then a direct computation shows that the following equations hold
\begin{equation}
\begin{split}
\label{minors}
M^{'}_{\frac{i_{1},i_{2},...,i_{k}}{1,2,...,k}}& = (-\sum_{l=1}^{k}a_{i_{l}i_{l}} + \sum_{j=1}^{k}\lambda_{j})
M_{\frac{i_{1},i_{2},...,i_{k}}{1,2,...,k}}\\
&\quad-2\sum_{s=1}^{k}\sum^{i_{s}-1}_{i=1}a_{ii_{s}}(-1)^{\sigma(p)}M_{\frac{i_{1},...,i,...\hat{i}_{s},...,i_{k}}{1,2,...,k}},
\end{split}
\end{equation}
where $i_{1}...i...\hat{i}_{s}...i_{k}$ is a permutation $p$ of the set $i_{1},...,i,...\hat{i}_{s},...,j_{k}$, (a hat above an element means, that this element is omitted) $\sigma(p)$ is the parity of $p$, i.e. the number of inversions in $p$. The indices $i_k$ in the sets $i_{1},i_{2},...,i_{k}$ and $i_{1},...,i,...\hat{i}_{s},...,i_{k}$ are increasing (recall, that the numerator $a$ in the index $\frac{a}{b}$ corresponds to the rows and $b$ to the columns of $\Psi$, and $a_{i_{r}i_{s}}$ is an entry of the Lax matrix $L$). To prove this formula consider the minor $M_{\frac{i_{1},i_{2},...,i_{k}}{1,2,...,k}}$. It is equal to the following polynomial
\beq{minors-det}
\begin{array}{c}
M_{\frac{i_{1},i_{2},...,i_{k}}{1,2,...,k}} = \sum_{p}(-1)^{\sigma(p)}\psi_{i_{1}j_{1}}\psi_{i_{2}j_{2}}...\psi_{i_{k}j_{k}},
\end{array}
\eq
where $j_{1}j_{2}...j_{k}$ is a permutation $p$ of the set $1,2,...k$ and $\sigma(p)$ is the number of inversions in this permutation. Differentiate this expression with respect to $t$ using the formula
\begin{align}
\label{psi-d-1}
\psi^{'}_{ij}&=(-a_{ii} + \lambda_{j})\psi_{ij}-2\sum^{i}_{k=1}a_{ki}\psi_{kj}.\\
\intertext{Change the subscripts so that the formula for dynamics $\psi$ takes the following form:}
\label{psi-d-2}
\psi^{'}_{i_{s}j_{s}}&=(-a_{i_{s}i_{s}} + \lambda_{j_{s}})\psi_{i_{s}j_{s}}-2\sum^{i_{s}-1}_{i=1}a_{ii_{s}}\psi_{ij_{s}};\\
\intertext{Finally take the formula}
\label{minors-2}
\begin{split}
M^{'}_{\frac{i_{1},i_{2},...,i_{k}}{1,2,...,k}}& = \sum_{p}(-1)^{t(p)}(\psi^{'}_{i_{1}j_{1}}\psi_{i_{2}j_{2}}...\psi_{i_{k}j_{k}}\\
&\qquad\quad+\psi_{i_{1}j_{1}}\psi^{'}_{i_{2}j_{2}}...\psi_{i_{k}j_{k}}+...+\psi_{i_{1}j_{1}}\psi_{i_{2}j_{2}}...\psi^{'}_{i_{k}j_{k}}),
\end{split}
\end{align}
substitute (\ref{psi-d-2}) in (\ref{minors-2}) and after regrouping we get
\begin{equation}
\label{minors-3}
\begin{split}
M^{'}_{\frac{i_{1},i_{2},...,i_{k}}{1,2,...,k}}& = \sum_{p}(-1)^{t(p)}((-\sum_{l=1}^{k}a_{i_{l}i_{l}} + \sum_{j=1}^{k}\lambda_{j})\psi_{i_{1}j_{1}}\psi_{i_{2}j_{2}}...\psi_{i_{k}j_{k}}\\
&\quad+(-2\sum^{i_{1}-1}_{i=1}a_{ii_{1}}\psi_{ij_{1}})\psi_{i_{2}j_{2}}...\psi_{i_{k}j_{k}}+\psi_{i_{1}j_{1}}(-2\sum^{i_{2}-1}_{i=1}a_{ii_{2}}\psi_{ij_{2}})\psi_{i_{3}j_{3}}...\psi_{i_{k}j_{k}}\\
&\quad+...+\psi_{i_{1}j_{1}}\psi_{i_{2}j_{2}}...\psi_{i_{k-1}j_{k-1}}(-2\sum^{i_{k}-1}_{i=1}a_{ii_{k}}\psi_{ij_{k}})),
\end{split}
\end{equation}
which is equivalent to \eqref{minors}.

Now formula \eqref{minors} shows that the derivative of $M_{\frac{i_{1},i_{2},...,i_{k}}{1,2,...,k}}$ can be expressed in the terms of minors $M_{\frac{i'_{1},i'_{2},...,i'_{k}}{1,2,...,k}}$, where at least one of the entries $i'_l$ is less than the corresponding $i_l$. Hence we can use induction to show that the equations $M(t)_{\frac{i_{1},i_{2},...,i_{k}}{1,2,...,k}}=0$ hold identically for all $t$.
\end{proof}

The following statement is a direct corollary of this lemma:
\begin{cor}
Toda flow preserves Schubert cells (in other words, the corresponding vector field is tangent to Schubert cells).
\end{cor}
In fact, the lemma shows that the system of equations which define the Schubert cell is preserved by the Toda flow: this is evident for the case, when the rank of the "vertical"\ submatrix differs from the maximal by $1$. In a more general case, we should apply this lemma to the matrix, equal to $\Psi$ with few omitted columns.

If we define the \textit{dual Schubert cells} as the sets of matrices for which the conditions, similar to the definition \ref{defischu1}, hold for the \textit{lower left} submatrices (see section \ref{sectorder}), then we obtain a similar statement
\begin{cor}
Toda flow preserves dual Schubert cells.
\end{cor}
Now we can use the general facts about the Schubert cells and dual Schubert cells, listed in section \ref{sectorder}.

We shall also use the following general proposition:
\begin{prop}
\label{propmf}
Let $(M,\,g)$ be a Riemannian manifold, $N$ be a manifold, immersed in $M$ via a smooth full rank map $F$. Let $f$ be a function on $M$. Let $g'$ and $f'$ be the restrictions of $g$ and $f$ to $N$. Then
$$
D_F(grad_{g'}f')=proj^\perp_F(N)(grad_gf),
$$
where $D_F$ is the differential of the map $F$.
\end{prop}
\begin{proof}
The statement of this proposition is local, it is enough to prove it for the tangent space at a point, so we can assume that $M=\mathbb R^a,\ N=\mathbb R^b$ and $F$ is a linear embedding. Choose orthonormal base for $g'$ in $\mathbb R^b$ and complete it to an orthonormal basis in $\mathbb R^a$ (with respect to $g$). Then in the corresponding coordinates $grad_gf$ is just the vector of partial derivatives $\left(\frac{\partial f}{\partial x^1},\dots,\frac{\partial f}{\partial x^a}\right)$ and $D_F(grad_{g'}f')=\left(\frac{\partial f}{\partial x^1},\dots,\frac{\partial f}{\partial x^b},0\dots,0\right)$, so that the statement is evident.
\end{proof}

We shall apply this statement to Schubert cells. One more statement that we shall need is the following:
\begin{prop}
\label{propdeMPe}
If the order of the eigenvalues in $\Lambda$ is given by $\lambda_1<\lambda_2<\dots<\lambda_n$, then the negative (resp. positive) eigenspace of the Hessian of the Morse function $F_n$ at $w$ is spanned by those positive roots of $SL_n$, which are mapped into negative (resp. positive) roots by $w$.
\end{prop}
This proposition (in a general case of arbitrary Cartan pair) was proven in the paper \cite{deMariPedroni}. For the sake of completeness we give here our prove of this fact in the case of the pair $Sl_n(\mathbb R)\supset SO_n(\mathbb R)$.

\begin{proof}
First of all let us fix some matrix $\tilde{s_k} \in \widetilde{S}_{n}$ (the definition of $\widetilde S_n$ is similar to the definitions of $\widetilde S_3$ and $\widetilde S_4$)
\beq{s_kn}  \tilde{s_k}= \left(
\begin{array}{c c c c c c}
 0 & \ast & 0 & ... &0\\
 0 & 0 & \ast &... & 0\\
 ... & ... & ... & ...& ...\\
 \ast & 0 & ... & ... &0\\
 0 & 0 & ... &  \ast & 0\\
\end{array}
\right).
\eq
The non-zero elements $\ast$ of $\tilde s_k$ are equal to $\pm1$. Conjugation with $\tilde s_k$ defines a permutation $w$ on the eigenvalues in the diagonal matrix $\Lambda$.

The matrix $\Psi$ in a neighborhood of $\tilde{s_k}$ has the following decomposition
\beq{decomppsi-2}
\Psi=\tilde{s_k}+\Theta\tilde{s_k}+\frac12\Theta^2\tilde{s_k}+o(\theta^2).
\eq
Here $o(\theta^2)$ denotes the sum of terms of degree 3 and higher in $\theta_{ij}$. Thus, we have
\beq{decomplax-2}
L(\Psi)=\Psi\Lambda\Psi^T=\tilde{s_k}\Lambda\tilde{s_k}^{-1}+ [\Theta, \tilde{s_k}\Lambda\tilde{s_k}^{-1}] - \Theta \tilde{s_k}\Lambda\tilde{s_k}^{-1} \Theta + \frac{1}{2}[\Theta^{2}, \tilde{s_k}\Lambda\tilde{s_k}^{-1}]_{+}+o(\theta^2),
\eq
where the matrix  $\Theta \in \mathfrak{so(n)}$ has the following form
\beq{angles}
\Theta= \left(
\begin{array}{c c c c c c}
 0 & \theta_{12} & ... & ... & \theta_{1n}\\
 -\theta_{12} & 0 & ... &... & \theta_{2n}\\
 ... & ... & ... & ...& ...\\
 -\theta_{1n-1} & -\theta_{2n-1} & ... & 0 & \theta_{n-1n}\\
 -\theta_{1n} & -\theta_{2n} & ... & -\theta_{n-1n} & 0\\
\end{array}
\right),
\eq

We shall be interested in
$$- \Theta \tilde{s_k}\Lambda\tilde{s_k}^{-1} \Theta + \frac{1}{2}[\Theta^{2}, \tilde{s_k}\Lambda\tilde{s_k}^{-1}]_{+} $$
to calculate the canonical form of the Hessian $d^2_{\tilde{s_k}}F_n$.
Consider the diagonal matrix $\widetilde{\Lambda_{w}}$
\beq{permut-lambda}
 \tilde{s_k}\Lambda\tilde{s_k}^{-1} = \widetilde{\Lambda_{w}}
\eq
then its matrix elements are $\widetilde{\lambda}_{ii}=\lambda_{w(i)}$. We need only diagonal parts of the terms $- \Theta \tilde{s_k}\Lambda\tilde{s_k}^{-1} \Theta$ and $+ \frac{1}{2}[\Theta^{2}, \tilde{s_k}\Lambda\tilde{s_k}^{-1}]_{+}$ because $F_n= Tr(LN)$, where the diagonal matrix $N$ has matrix elements $N_{ii}=i-1$. Put
\begin{align*}
\varphi^{1}=&- \Theta \tilde{s_k}\Lambda\tilde{s_k}^{-1} \Theta, & \varphi^{1}_{ii}&=\sum_{j \neq i} \theta_{ij}^{2}\lambda_{w(j)},\\
\varphi^{2}=&+ \frac{1}{2}[\Theta^{2}, \tilde{s_k}\Lambda\tilde{s_k}^{-1}]_{+}, & \varphi^{2}_{ii}&=-\sum_{j \neq i} \theta_{ij}^{2}\lambda_{w(i)},
\end{align*}
for $j=1,\dots,n$. Here we replace $\theta_{ij}$ by $\theta_{ji}$ if $i>j$. We compute
\begin{equation}
\label{hesscalc}
\begin{split}
d^2_{\tilde{s_k}}F_n&= Tr((\varphi^{1}+\varphi^{2})N)=\\
 &=\sum_{i<j} \theta_{ij}^{2}(\lambda_{w(j)}-\lambda_{w(i)})(i-1)+ \sum_{i<j} \theta_{ij}^{2}(\lambda_{w(i)}-\lambda_{w(j)})(j-1)=\\
 &= \sum_{i<j} \theta_{ij}^{2}(\lambda_{w(i)}-\lambda_{w(j)})(j-i).
\end{split}
\end{equation}
Thus we have expressed by direct calculation the quadratic approximation of $F_n$ near $\tilde{s_k}$ in coordinates $\theta_{ij}$. The statement of the lemma now follows from a direct inspection of this formula.
\end{proof}

Finally, let us put together all the statements: it follows from the propositions \ref{propmf} and \ref{propdeMPe} and the second property of the Schubert cells (see section \ref{sectorder}), that Schubert cell $X_w$ (resp. the dual Schubert cell $\Omega_w$) coincides with the unstable (resp. stable) submanifold of the Morse function $F_n$ at $w$. Indeed, from proposition \ref{propmf} it follows that the Morse flow of the restriction ${F_n}|_{X_w}$ (resp. ${F_n}|_{\Omega_u}$) coincides with the Toda flow. And from proposition \ref{propdeMPe} it follows, that $w$ is local maximum (resp. minimum) on $X_w$ (resp. $\Omega_w$). Now the proposition follows from the first property (see section \ref{sectorder}) of the Schubert cells and their duals.
\end{proof}

Finally observe, that we have actually showed, that Toda system is in fact a Morse-Smale system (see the definition \ref{MoSm}). Thus we obtain one more statement:
\begin{cor}
\label{corMS}
The dimension of the space, spanned by the trajectories connecting two singular points, is equal to the distance between these points in the Hasse graph of the Bruhat order.
\end{cor}
\begin{proof}
The statement follows from the transversality of the Schubert and dual Schubert cells and the fact, that the dimension of the Schubert (resp. dual Schubert) cell of $w$ is equal to the number of the elements, connecting it to the bottom (resp. top) element in the Bruhat diagramm.
\end{proof}

\begin{rem}\rm
In effect, as one sees, the asymptotics of the trajectory through a point of the flag space is determined by the Schubert cell and the dual Schubert cell, which contain the point. Namely, if the point belongs to the intersection of the dual Schubert cell $\Omega_\alpha$ and the Schubert cell $X_\beta$, then the trajectory of the Toda flow through the point goes from $\alpha$ to $\beta$.
\end{rem}

\begin{rem}\rm
Another approach to the question of describing the phase portrait of Toda flow was given in the paper \cite{FS}, where the authors use explicit formulas for the solutions to describe their asymptotic behaviour. One of the outcomes of their approach is a construction of a polytope, whose vertices (after suitable identifications) correspond to the elements of the Weyl group. One may ask, if the edges of this polytope reproduce the structure of the minimal edges of Bruhat order (one can see, that this is so for $SO_3(\mathbb R)$ case, which is proved in the cited paper).
\end{rem}

\begin{center}
\large{\bf Appendix}
\end{center}
\appendix
\section{Geometry of the Toda flow on $SO(3)$}
In this additional section we try to visualize the results, concerning the lowest dimensional case we considered before, that of $SO(3)$.

This group contains a normal subgroup $N=N_3$ of diagonal matrices (the first row of the list \eqref{classesS3}), which commutes with the matrix $\Lambda$. It follows, that the vector field $A(\Psi)$ does not depend on the choice of element in the right $N$\/-coset of $SO_3(\mathbb R)$. So we consider the Toda system as an evolution of points on factor-space $SO_3(\mathbb R)/N$. The singular points of this system in this setting will consist of $6$ cosets ($\tilde{s_i},\ i=1,\dots,6$), which we regard as the elements of the permutation group $S_3=W(\mathfrak{sl}_3)$. One can visualize the structure of this space and positions of the singular points on it, using the quaternionic presentation of $SO_3(\mathbb R)$:
$$SO_3(\mathbb R)= \mathbb H_1 / (\mathbb Z/2\mathbb Z).$$
Here $\mathbb H_1$ denotes the unit quaternionic sphere and $\mathbb Z\left/2\mathbb Z\right.$ acts on it by $\pm1$. Then the factor space $SO_3(\mathbb R)/N$ is equal to the factorization of $S^3$ by the preimage $\tilde N$ of $N$ in $\mathbb H$. One can find the explicit quaternions, representing the elements of $\tilde N$:
\begin{align*}
{}& \pm 1, & {}&\pm\mathbf i, & {}&\pm\mathbf j, & {}&\pm\mathbf k,
\end{align*}
then the factor-space $\mathbb H_1/\tilde N$ can be identified with the factorisation of the spherical Voronoy set of $1\in\mathbb H$ (with respect to the other elements of $\tilde N$) along its boundary, see figure \ref{fig:Graph-6_1}. Recall, that the \textit{spherical Voronoy set} of a point $x$ in a finite set $X\subset S^3$ consists of all points of $S^3$, which lie closer to $x$ than to any othe point in $X$ with respect to the spherical distance.
\begin{figure}[!t]
\begin{center}
\includegraphics[scale=.45]{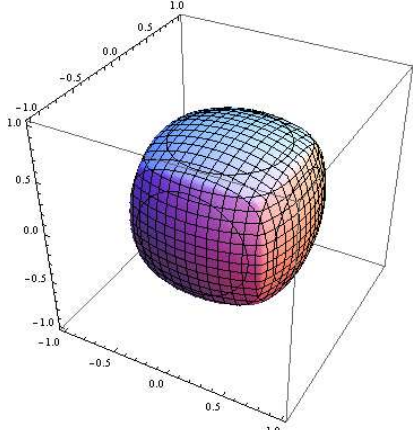}
	\caption{The Voronoy set of $\tilde N$}
	\label{fig:Graph-6_1}
\end{center}
\end{figure}
As one sees, this domain is (combinatorically) a cube. The action of $\tilde N$ identifies the opposite faces of this cube after having turned them $90^\circ$ around their centers. It is also easy to describe the positions of the singular points of $A(\psi)$ at this set, see figure \ref{fig:fig3}: at this figure there are (after identifications) shown all $6$ singular points: $1$ point is in the center of the cube, $3$ points are the centers of the three pairs of opposite faces of the cube, and $2$ more points are given by the orbits of the cube's vertices under the action, which identifies the opposite faces, turned by $90^\circ$.
\begin{figure}[!t]
\begin{center}
		\includegraphics[scale=.85]{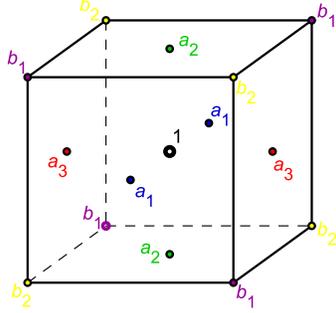}
	\caption{The fundamental domain of $N$, positions of the singular points}
	\label{fig:fig3}
\end{center}
\end{figure}

There are six Schubert cells $X_w$ in this case and we can write down the systems of equations, that determine each of them. In fact, we have a cellular decomposition
\beq{decomposition}
FL_{3}(\mathbb R) = \bigcup_{w \in \widetilde S_3} X_{w}
\eq
of the flag manifold. As we have explained in section \ref{sectorder}, we can describe the preimages of these cells in $SL_3(\mathbb R)$ with the help of systems of equations on matrix entries (these preimages are often called \textit{the matrix Schubert varieties}). Here we describe the intersections of these preimages with $SO_3(\mathbb R)\subset SL_3(\mathbb R)$.

For each permutation $w_{i}$ we have the orthogonal matrix $\tilde{s_{i}}$:
\begin{align*}
\tilde{s_{1}}:\begin{pmatrix} \ast & 0 & 0\\
                0 & \ast & 0 \\
                0 & 0 & \ast \end{pmatrix}, \ \ \
\tilde{s_{2}}:\begin{pmatrix} 0 & \ast & 0\\
                \ast & 0 & 0 \\
                0 & 0 & \ast \end{pmatrix}, \ \ \
\tilde{s_{3}}:\begin{pmatrix} 0 & 0 & \ast\\
                0 & \ast & 0 \\
                \ast & 0 & 0 \end{pmatrix}, \\
\tilde{s_{4}}:\begin{pmatrix} \ast & 0 & 0\\
                0 & 0 & \ast \\
                0 & \ast & 0 \end{pmatrix}, \ \ \
\tilde{s_{5}}:\begin{pmatrix} 0 & \ast & 0\\
                0 & 0 & \ast \\
                \ast & 0 & 0 \end{pmatrix}, \ \ \
\tilde{s_{6}}:\begin{pmatrix} 0 & 0 & \ast\\
                \ast & 0 & 0 \\
                0 & \ast & 0 \end{pmatrix},
\end{align*}
where $\ast$ takes the values $\pm1$. The corresponding rank matrices and polynomials cutting out the matrix Schubert varieties inside $SO_{3}(\mathbb R)$ have the following form:
\begin{align*}
rk(\tilde{s_{1}})&=\begin{pmatrix} 1 & 1 & 1\\
                1 & 2 & 2 \\
                1 & 2 & 3 \end{pmatrix}, \ \mathrm{no \ polynomials}, \\
rk(\tilde{s_{2}})&=\begin{pmatrix} 0 & 1 & 1\\
                1 & 2 & 2 \\
                1 & 2 & 3 \end{pmatrix}, \ \psi_{11}=0,\\
rk(\tilde{s_{3}})&=\begin{pmatrix} 0 & 0 & 1\\
                0 & 1 & 2 \\
                1 & 2 & 3 \end{pmatrix}, \ \psi_{11}=0,\ \psi_{12}=0,\ \psi_{21}=0,\ M_{\frac{12}{12}}=0,\\
rk(\tilde{s_{4}})&=\begin{pmatrix} 1 & 1 & 1\\
                1 & 1 & 2 \\
                1 & 2 & 3 \end{pmatrix},  \ M_{\frac{12}{12}}=0,\\
rk(\tilde{s_{5}})&=\begin{pmatrix} 0 & 1& 1\\
                0 & 1 & 2 \\
                1 & 2 & 3 \end{pmatrix},  \ \psi_{11}=0,\ \psi_{21}=0,\ M_{\frac{12}{12}}=0,\\
rk(\tilde{s_{6}})&=\begin{pmatrix} 0 & 0 & 1\\
                1 & 1 & 2 \\
                1 & 2 & 3 \end{pmatrix}, \ \psi_{11}=0,\ \psi_{12}=0,\ M_{\frac{12}{12}}=0.
\end{align*}
As one can see, in many cases the recipe, that we use, produces more equations than one really needs to describe the cell. Thus the dimensions of the the cells should be calculated with care: this list does not give them "on the nose".

\bigskip
\addcontentsline{toc}{section}{\numberline{}Acknowledgments}
\paragraph{Acknowledgments}
The work of Yu.B. Chernyakov was supported in part by grants RFBR-12-02-00594 and by the Federal Agency for Science and Innovations of Russian Federation under contract 14.740.11.0347. The work of G.I. Sharygin was supported by the Federal Agency for Science and Innovations of Russian Federation under contract 14.740.11.0081. The work of A.S. Sorin was supported in part by the RFBR Grants No. 11-02-01335-a and No. 11-02-12232-ofi-m-2011.


\begin{thebibliography}{60}

\bibitem{T1}
M. Toda, Vibration of a chain with nonlinear interaction, J. Phys. Soc. Japan 22(2) (1967), 431 -- 436.

\bibitem{T2}
M. Toda, Wave propagation in anharmonic lattices, J. Phys. Soc. Japan 23(3) (1967), 501 -- 506.

\bibitem{H}
M. Henon, Integrals of the Toda lattice, Phys. Rev. B9 (1974), 1921 -- 1923.

\bibitem{F1}
H. Flaschka, The Toda lattice. I. Existence of integrals, Phys. Rev. B 9(4) (1974), 1924 -- 1925.

\bibitem{F2}
H. Flaschka, On the Toda lattice. II. Prog. Theor. Phys. 51(3) (1974), 703 -- 716.

\bibitem{Mos}
J. Moser. Finitely Many Mass Points on the Line Under the Influence
of an Exponential Potential - An Integrable System. in: Dynamical Systems,
Theory and Applications, J. Moser ed. Lec. Notes in Phys., Vol.38,
Springer-Verlag, Belrin and New York (1975), pp. 467-497.

\bibitem{Mos2} J. Moser. Integrable hamiltonian systems and spectral theory. Pisa:
Lezioni Fermiane 1981.

\bibitem{A}
А. Архангельский, Вполне интегрируемые системы на группе треугольных матриц, Мат.сборник, т.108(150), №1 (1979), 134 -- 142.
A.A. Arhangelskii, Completely integrable hamiltonian systems on a group of triangular matrices, Mathematics of the USSR-Sbornik (1980), 36:1, 127 -- 134.

\bibitem{Ad}
M. Adler, On a trace functional for pseudo-differential operators and the symplectic structure of the Korteweg-de Vries equation, Invent. Math., 50 (1979), 219 -- 248.

\bibitem{K1}
B. Kostant, The solution to a generalized Toda lattice and representation theory, Adv. in Math. 34 (1979), 195 -- 338.

\bibitem{S1}
W. W. Symes, Systems of Toda type, inverse spectral problems, and representation theory, Invent. Math. 59 (1980), no. 1, 13 -- 51.

\bibitem{DNT}
P. Deift, T. Nanda, and C. Tomei. Ordinary Differential Equations and
the Symmetric Eigenvalue Problem. SIAM J. Numer. Anal., 20 (1983), 1 -- 20.

\bibitem{DLNT}
P. Deift, L. C. Li, T. Nanda, and C. Tomei, The Toda flow on a generic orbit is integrable, CPAM 39 (1986), 183 -- 232.

\bibitem{EFS}
N. Ercolani, H. Flaschka, and S. Singer, The geometry of the full Kostant-Toda lattice In: Integrable Systems,
Vol. 115 of Progress in Mathematics, Birkhauser (1993), 181 -- 226.

\bibitem{BBR}
A. M. Bloch, R. W. Brockett, T. S. Ratiu, A new formulation of the generalized Toda lattice equations and their fixed point analysis via the momentum map, Bull. Amer. Math. Soc. (N.S.) Volume 23, Number 2 (1990), 477 -- 485.

\bibitem{BG}
A. M. Bloch and M. Gekhtman, Hamiltonian and gradient structures in the Toda flows, J. Geom. Phys. 27 (1998), 230 -- 248

\bibitem{deMariPedroni}
F. De Mari, M. Pedroni, Toda flows and real Hessenberg manifolds. J. Geom. Anal., 9 no.4 (1999), 607 -- 625.

\bibitem{KM}
Y. Kodama and K. T-R McLaughlin, Explicit Integration of the Full Symmetric Toda Hierarchy and the Sorting Property, Lett. Math. Phys., 37 (1996), 37 -- 47.

\bibitem{F}
W. Fulton, Young Tableaux, Cambridge University Press, 1997.

\bibitem{UW}
H. Ulfarsson, A. Woo, Which Schubert varieties are local complete intersections?, arXiv:1111.6146 (to appear in the proceedings of LMS).

\bibitem{SBVL}
S. Billey and V. Lakshmibai, Singular loci of Schubert varieties, Progr. Math. 182 (2000), Birkhauser,. Boston.

\bibitem{FS}
P. Fre, A.S. Sorin, The arrow of time and the Weyl
group: all supergravity billiards are integrable, Nucl. Phys., B 815 (2009), 430,[arXiv:0710.1059];
Integrability of supergravity billiards and the generalized Toda lattice equation, Nucl. Phys., B 733 (2006), 334, [hep-th/0510156].

\bibitem{CS}
Yu.B. Chernyakov, A.S. Sorin, Integrability of the Full Symmetric Toda system, to appear.

\end{thebibliography}
\end{document}